\documentclass[preprint,numbers]{sigplanconf}

\usepackage{amsmath}

\usepackage{amssymb}
\usepackage{tikz}
\usetikzlibrary{automata}
\usepackage[boxed,linesnumbered,longend]{algorithm2e}
\usepackage{url}
\usepackage{subcaption}

\usepackage{QED}

\newtheorem{definition}{Definition}
\newtheorem{proposition}[definition]{Proposition}
\newtheorem{lemma}[definition]{Lemma}
\newtheorem{theorem}[definition]{Theorem}
\newtheorem{corollary}[definition]{Corollary}

\newtheorem{observation}[definition]{Observation}
\newtheorem{assumption}[definition]{Assumption}

\newcommand{\e}{\epsilon}
\newcommand{\p}{\mathbb{P}}
\newcommand{\N}{\mathbb{N}}

\newcommand{\eqclass}[1]{[#1]}

\begin{document}

\setlength{\pdfpageheight}{\paperheight}
\setlength{\pdfpagewidth}{\paperwidth}

\conferenceinfo{CONF 'yy}{Month d--d, 20yy, City, ST, Country}
\copyrightyear{20yy}
\copyrightdata{978-1-nnnn-nnnn-n/yy/mm}
\copyrightdoi{nnnnnnn.nnnnnnn}


\titlebanner{banner above paper title}        
\preprintfooter{Betz, Le Roux, Stable states of perturbed MC}   

\title{Stable states of Perturbed Markov Chains}

\authorinfo{Volker Betz}
           {Technische Universit\"at Darmstadt}
           {betz@mathematik.tu-darmstadt.de}
\authorinfo{St\'ephane Le Roux\thanks{This author is supported by the ERC inVEST (279499) project and was in TU Darmstadt when this research started.}}
           {Universit\'e libre de Bruxelles}
           {stephane.le.roux@ulb.ac.be}

\maketitle

\begin{abstract}
\noindent Given an infinitesimal perturbation of a discrete-time finite Markov chain, we seek the states that are stable despite the perturbation, \textit{i.e.} the states whose weights in the stationary distributions can be bounded away from $0$ as the noise fades away. Chemists, economists, and computer scientists have been studying irreducible perturbations built with exponential maps. Under these assumptions, Young proved the existence of and computed the stable states in cubic time. We fully drop these assumptions, generalize Young's technique, and show that stability is decidable as long as $f\in O(g)$ is. Furthermore, if the perturbation maps (and their multiplications) satisfy $f\in O(g)$ or $g\in O(f)$, we prove the existence of and compute the stable states and the metastable dynamics at all time scales where some states vanish. Conversely, if the big-$O$ assumption does not hold, we build a perturbation with these maps and no stable state. Our algorithm also runs in cubic time despite the general assumptions and the additional work. Proving the correctness of the algorithm relies on new or rephrased results in Markov chain theory, and on algebraic abstractions thereof.
\end{abstract}

\category{G}{3}{}

\terms


\section{Introduction}\label{sect:intro}

Motivated by the dynamics of chemical reactions, Eyring~\cite{Eyring35} and Kramers~\cite{Kramers40} studied how infinitesimal perturbations of a Markov chain affects its stationary distributions. This topic has been further investigated by several academic communities including probability theorists, economists, and computer scientists. In several fields of application, such as learning and game theory, it is sometimes unnecessary to describe the exact values of the limit stationary distributions: it suffices to know whether these values are zero or not. Thus, the \emph{stochastically stable states} (\cite{FY90},~\cite{KMR93},~\cite{Young93}) were defined in different contexts as the states that have positive probability in the limit. We rephrase a definition below.

\begin{definition}[Markov chain perturbation and stochastic stability]\label{defn:mcp-ss}
Let $I$ be a subset of positive real numbers with $0$ as a limit point for the usual topology\footnote{This implies that $I$ is infinite. $]0,1]$ and $\{\frac{1}{2^n}\,|\,n\in \mathbb{N}\}$ are typical $I$.}. A perturbation is a family $((X_n^{(\e)})_{n\in\N})_{\e\in I}$ of discrete-time Markov chains sharing the same finite state space. If the chain $(X_n^{(\e)})_{n\in\N}$ is irreducible for all $\e \in I$, then $((X_n^{(\e)})_{n\in\N})_{\e\in I}$ is said to be an irreducible perturbation.

A state $x$ of $((X_n^{(\e)})_{n\in\N})_{\e\in I}$ is stochastically stable if there exists a family of corresponding stationary distributions $(\mu_\e)_{\e\in I}$ such that $\liminf_{\e\to 0} \mu_{\e}(x) > 0$. It is stochastically fully vanishing if $\limsup_{\e\to 0} \mu_{\e}(x) = 0$ for all $(\mu_\e)_{\e\in I}$. Non-stable states are called vanishing.
\end{definition}

\noindent Definition~\ref{defn:mcp-ss} may be motivated in at least two ways. First, a dynamical system (\textit{e.g.} modeled by a Markov chain) has been perturbed from the outside, and the laws governing the systems (\textit{e.g} the transition probability matrix) have been changed. As time elapses (\textit{i.e.} as $\e$ approaches zero), the laws slowly go back to normal. What are the almost sure states of the system after infinite time? Second, a very complex Markov chain is the sum of a simple chain and a complex perturbation matrix that is described \textit{via} a small, fixed $\e_0$. The stationary distributions of the complex chain are hard to compute, but which states have significantly positive probability after infinite time? Our main result below answers these questions.

\begin{theorem}\label{thm:teaser}
Consider a perturbation such that $f \in O(g)$ or $g \in O(f)$ for all $f$ and $g$ in the multiplicative closure of the transition probability functions $\e \mapsto p_\e(x,y)$ with $x \neq y$. Then the perturbation has stable states, and stability can be decided in $O(n^3)$, where $n$ is the number of states.
\end{theorem}

Note that by finiteness of the state space it is easy to prove that every perturbation has a state that is not fully vanishing.

\subsection{Related works and comparisons}

In 1990 Foster and Young~\cite{FY90} defined the stochastically stable states of a general (continuous) evolutionary process, as an alternative to the evolutionary stable strategies~\cite{MP73}. Stochastically stable states were soon adapted by Kandori, Mailath, and Rob~\cite{KMR93} for evolutionary game theory with $2\times 2$ games. Then Young~\cite[Theorem 4]{Young93} proved "a finite version of results obtained by Freidlin and Wentzel" in~\cite{FW98}. Namely, he characterized the stochastically stable states if the perturbation satisfies the following assumptions: 1) the perturbed matrices $P^{\e}$ are aperiodic and irreducible; 2) the $P^\e$ converge to the unperturbed matrix $P^0$ when $\e$ approaches zero; 3) every transition probability is a function of $\e$ that is equivalent to $c \cdot \e^{\alpha}$ for some non-negative real numbers $c$ and $\alpha$. The main tool in Young's proof was proved by Kohler and Vollmerhaus~\cite{KV80} and is the special case for irreducible chains of the Markov chain tree theorem (see~\cite{LR83} or~\cite{FW98}). Young's characterization involves minimum directed spanning trees, which can be computed in $O(n^2)$~\cite{GGST86} for graphs with $n$ vertices. Since there are at most $n$ roots for directed spanning trees in a graph with $n$ vertices, Young can compute the stable states in $O(n^3)$.

In 2000, Ellison~\cite{Ellison00} characterized the stable states \textit{via} the alternative notion of the radius of a basin of attraction. The major drawback of his characterization compared to Young's is that it is "not universally applicable"~\cite{Ellison00}; the advantages are that it provides "a bound on the convergence rate as well as a long-run limit" and "intuition for why the long-run stochastically stable set of a model is what it is". 
In 2005, Wicks and Greenwald~\cite{WG05} designed an algorithm to express the exact values of the limit stationary distribution of a perturbation, which, as a byproduct, also computes the set of the stable states. Like~\cite{Young93} they consider perturbations that are related to the functions $\e\mapsto \e^{\alpha}$, but they only require that the functions converge exponentially fast. Also, instead of requiring that the $P^\e$ be irreducible for $\e > 0$, they only require that they have exactly one essential class. They do not analyze the complexity of their algorithm but it might be polynomial time. We improve upon~\cite{Young93},~\cite{Ellison00}, and~\cite{WG05} in several ways.
\begin{enumerate}
\item\label{improve1} The perturbation maps in the literature relate to the maps $\e\mapsto \e^{\alpha}$. Their specific form and their continuity, especially at $0$, are used in the existing proofs. Theorem~\ref{thm:teaser} dramatically relaxes this assumption. Continuity, even at $0$, is irrelevant, which allows for aggressive, \textit{i.e.}, non-continuous perturbations. We show that our assumption is (almost) unavoidable.

\item The perturbations in the literature are irreducible (but \cite{WG05} slightly weakened this assumption). It is general enough for perturbations relating to the maps $\e \mapsto \e^\alpha$, since it suffices to process each sink (aka bottom) irreducible component independently, and gather the results. Although this trick does not work for general perturbation maps, Theorem~\ref{thm:teaser} manages not to assume irreducibility.


\item The perturbation is abstracted into a weighted graph and shrunk by combining recursively a shortest-path algorithm (w.r.t. some tropical-like semiring) and a strongly-connected-component algorithm. Using tropical-like algebra to abstract over Markov chains has already been done before, but not to solve the stable state problem. (\cite{GKMS15} did it to prove an algebraic version of the Markov chain tree theorem.)

\item Our algorithm computes the stable states in $O(n^3)$, as in~\cite{Young93}, which is the best known complexity. In addition, the computation itself is a summary of the asymptotic behavior of the perturbation: it says at which time scales the vanishing states vanish, and the intermediate graph obtained at each recursive stage of the algorithm accounts for the metastable dynamics of the perturbation at this vanishing time scale.
\end{enumerate}

\noindent Section~\ref{sect:nota} sets some notations; Section~\ref{sect:tga} analyses which assumptions are relevant for the existence of stable states; Section~\ref{sect:pp-ess} proves the existential part of Theorem~\ref{thm:teaser}, \textit{i.e.} it develops the probabilistic machinery to prove the existence of stable states; hinging on this, Section~\ref{sect:aqa} proves the algorithmic part of Theorem~\ref{thm:teaser}, \textit{i.e.} it abstracts the relevant objects using a new algebraic structure, presents the algorithm, and proves its correctness and complexity; Section~\ref{sect:disc} discusses two important special cases and an induction proof principle related to the termination of our algorithm.

\subsection{Notations}\label{sect:nota}

\begin{itemize}
\item The set $\mathbb{N}$ of the natural numbers contains $0$. For a set $S$ and $n\in\mathbb{N}$, let $S^n$ be the words $\gamma$ over $S$ of length $|\gamma| = n$. Let $S^* := \cup_{n\in\mathbb{N}}S^n$ be the finite words over $S$. The set-theoretical notation $\cup E := \cup_{x\in E}x$ is used in some occasions.
\item Let $(X_n)_{n\in\N}$ be a Markov chain with state space $S$. For all $A\subseteq S$ let $\tau_A := \inf \{n \geq 0: X_n \in A\}$ ($\tau^+_A := \inf \{n > 0: X_n \in A\}$)
be the first time (first positive time) that the chain hits a state inside $A$. Usually $\tau_{\{x\}}$ and $\tau^+_{\{x\}}$ are written $\tau_x$ and $\tau^+_x$, respectively.
\item Given a Markov chain $(X_n)_{n\in\N}$, the corresponding matrix representation, law of the chain when started at state $x$, expectation when started at state $x$, and possible stationary distributions are respectively denoted $p$, $\p^x$, $\mathbb{E}^x$, and $\mu$. When considering other Markov chains $(\tilde{X}_n)_{n\in\N}$ or $(\widehat{X}_n)_{n\in\N}$, the derived notions are denoted with tilde or circumflex, as in $\tilde{p}$ or $\widehat{\mu}$.
\item A perturbation $((X_n^{(\e)})_{n\in\N})_{\e\in I}$ will often be denoted $X$ for short, and when it is clear from the context that we refer to a perturbation, $p$ will denote the function $(\e,x,y)\mapsto p_\e(x,y)$ (instead of $(x,y)\mapsto p(x,y)$), and $p(x,y)$ will denote $\e\mapsto p_\e(x,y)$ (instead of a mere real number). The other derived notions are treated likewise.
\item The probability of a path is defined inductively by $p(xy) := p(x,y)$ and $p(xy\gamma) := p(x,y)p(y\gamma)$ for all $x,y\in S$ and $\gamma\in S\times S^*$.
\item Given $x$, $y$, and a set $A$, a simple $A$-path from $x$ to $y$ is a repetition-free (unless $x = y$) word $\gamma$ starting with $x$ and ending with $y$, and using beside $x$ and $y$ only elements in $A$. Formally, $\Gamma_A(x,y) := \{\gamma\in \{x\}\times A^* \times\{y\} \,|\, (1 \leq i < j \leq |\gamma| \wedge \gamma_i = \gamma_j) \Rightarrow (i = 1 \wedge j = |\gamma|)\}.$
\end{itemize}

\subsection{Towards general assumptions}\label{sect:tga}

\noindent A state $x$ of a perturbation is stable if there exists a related family $(\mu_\e)_{\e\in I}$ of stationary distributions such that $\mu(x) = O(1)$, but even continuous perturbations that converge when $\e$ approaches $0$ may fail to have stable states. For instance let $S := \{x,y\}$ and for all $\e \in]0,1]$ let $p_\e(x,y) := \e^2$ and $p_\e(y,x) := \e^{2+\cos(\e^{-1})}$ as in Figure~\ref{fig:no-stable1}, where the self-loops are omitted. In the unique stationary distribution $x$ has a weight $\mu_\e(x) = (1+\e^{-\cos(\e^{-1})})^{-1}$. Since $\mu_{(2n\pi)^{-1}}(x) = \frac{2n\pi}{1 + 2n\pi} \to_{n \to \infty} 1$ and $\mu_{(2(n+1)\pi)^{-1}}(x) = \frac{1}{1 + 2(n+1)\pi} \to_{n \to \infty} 0$, neither $x$ nor $y$ is stable.

As mentioned above, the perturbations in the literature are related to the functions $\e \mapsto \e^\alpha$ with $\alpha \geq 0$, which rules out the example from Figure~\ref{fig:no-stable1} and implies the existence of a stable state~\cite{Young93}. Here, however, we want to assume as little as possible about the perturbations, while still guaranteeing the existence of stable states. Towards it let us first rephrase the big $O$ notation as a binary relation. It is well-known that big $O$ enjoys various algebraic properties. The ones we need are mentioned in the appendix.

\begin{figure}
\centering
\begin{subfigure}[b]{0.15\textwidth}
\begin{tikzpicture}[shorten >=1pt,node distance=1.6cm, auto]
  \node[state] (x) {$x$};
  \node[state] (y) [right of = x] {$y$};

\path[->] (x) edge [bend right] node [below] {$\e^2$} (y)		
		(y) edge [bend right] node [above] {$\e^{2+\cos(\e^{-1})}$} (x) ;
 \end{tikzpicture}
 \caption{}
 \label{fig:no-stable1}
\end{subfigure}
\begin{subfigure}[b]{0.15\textwidth}
\begin{tikzpicture}[shorten >=1pt,node distance=1.5cm, auto]
  \node[state] (y) {$y$};
  \node[state] (z) [right of = y] {$z$};
 \node[state] (x) [above of = z]{$x$};

\path[->] (x) edge [bend right] node [above left] {$\e^6$} (y)
		(z) edge node [right] {$\e^4$} (x)
		(z) edge [bend right] node [above] {$1 - \e^4$} (y)		
		(y) edge [bend right ] node [below] {$\e^{2+\cos(\e^{-1})}$} (z) ;
 \end{tikzpicture}
  \caption{}
  \label{fig:no-stable2}
 \end{subfigure}
\begin{subfigure}[b]{0.1\textwidth}
\begin{tikzpicture}[shorten >=1pt,node distance=1.5cm, auto]
  \node[state] (x) {$x$};
  \node[state] (y) [right of = x] {$y$};
    \node[state] (z) [above of = y] {$z$};

\path[->] (x) edge [bend left] node [above] {$\frac{1+\cos (\e^{-1})}{2}$} (y)		
		(y) edge [bend left] node [below] {$\e$} (x)
		(z) edge[bend left] node {$\frac{1}{2}$} (y);
 \end{tikzpicture}
 \caption{}
 \label{fig:no-stable3}
\end{subfigure}

\begin{subfigure}[b]{0.45\textwidth}
\begin{tikzpicture}[shorten >=1pt,node distance=2.4cm, auto]
  \node[state] (x1) {$x_1$};
  \node[state] (x2) [right of = x1] {$x_2$};
  \node[state] (xn1) [right of = x2] {$x_{n-1}$};
  \node[state] (xn) [right of = xn1] {$x_n$};
  \node[state] (y1) [below of = xn] {$y_1$};
  \node[state] (y2) [left of = y1] {$y_2$};
  \node[state] (ym1) [left of = y2] {$y_{m-1}$};
  \node[state] (ym) [left of = ym1] {$y_m$};    

\path[->] (x1) edge node {$f_1$} (x2)		
		 (xn1) edge node {$f_{n-1}$} (xn)
		 (xn) edge node {$f_n$} (y1)		 
		 (y1) edge node {$g_1$} (y2)
		 (ym1) edge node {$g_{m-1}$} (ym)		 
		 (ym) edge node {$g_m$} (x1)
		 (x1) edge [loop above] node {$1-f_1$} ()
		 (x2) edge [bend left = 20] node {$1-f_2$} (x1)
		 (xn1) edge [bend left = 50] node [above] {$1-f_{n-1}$} (x1)
		 (xn) edge [bend right] node {$1-f_n$} (x1)
		 (y1) edge [loop below] node {$1-g_1$} ()		 
		 (y2) edge [bend left = 20] node {$1-g_2$} (y1)
		 (ym1) edge [bend left = 50] node [below] {$1-g_{m-1}$} (y1)
		 (ym) edge [bend right] node {$1-g_m$} (y1);
\path[dashed] (x2) edge node {} (xn1)
		 (y2) edge node {} (ym1);		 
 \end{tikzpicture}
 \caption{}
 \label{fig:no-stable4}
\end{subfigure}
 \caption{Perturbations without stable states}
 \label{fig:no-stable}
 \end{figure}
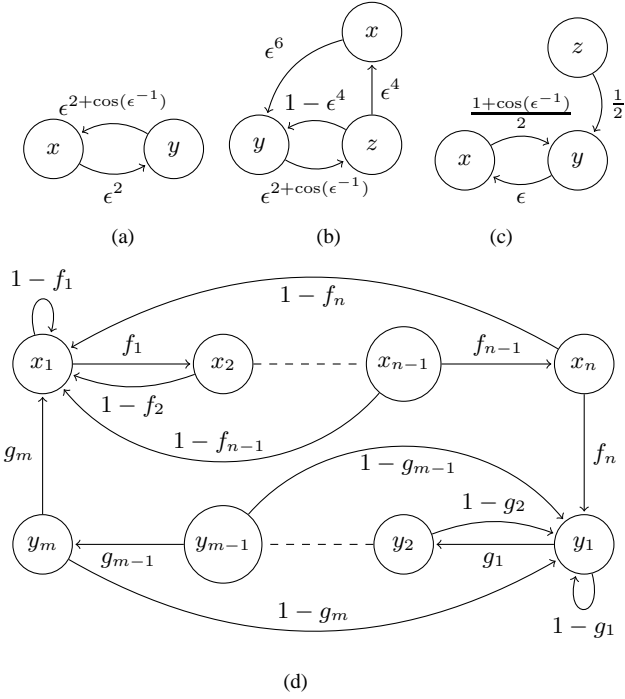

  \begin{figure}
\centering
\begin{subfigure}[b]{0.15\textwidth}
\begin{tikzpicture}[shorten >=1pt,node distance=1.2cm, auto]
  \node[state] (x) {$x$};
  \node[state] (y) [right of = x] {$y$};
    \node[state] (z) [above of = y] {$z$};

\path[->] (x) edge [bend left] node [below] {$\e$} (y)		
		(y) edge [bend left] node [below] {$\e^2$} (x)
		(z) edge node [left]{$\frac{1-\e}{3}$} (y);
 \end{tikzpicture}
  \caption{}
  \label{fig:stable1}
 \end{subfigure}
\begin{subfigure}[b]{0.15\textwidth}
\begin{tikzpicture}[shorten >=1pt,node distance=1.8cm, auto]
  \node[state] (x) {$x$};
  \node[state] (y) [right of = x] {$y$};

\path[->] (x) edge [bend right] node [below] {$\e(2-\cos(\e^{-1}))$} (y)		
		(y) edge [bend right] node [above] {$\e(2+\cos(\e^{-1}))$} (x) ;
 \end{tikzpicture}
 \caption{}
 \label{fig:stable2}
\end{subfigure}
\begin{subfigure}[b]{0.15\textwidth}
\begin{tikzpicture}[shorten >=1pt,node distance=1.5cm, auto]
  \node[state] (x) {$x$};
  \node[state] (y) [right of = x] {$y$};
  \node[state] (z) [above of = y] {$z$};

\path[->] (x) edge [bend left] node [above] {$\frac{1+\cos (\e^{-1})}{2}$} (y)		
		(y) edge [bend left] node [below] {$\frac{1+\cos (\e^{-1})}{2}$} (x)
		(z) edge [bend left]node {$1$} (y);
 \end{tikzpicture}
 \caption{}
 \label{fig:stable3}
\end{subfigure}
 \caption{Perturbations with stable states}
 \label{fig:stable}
 \end{figure}
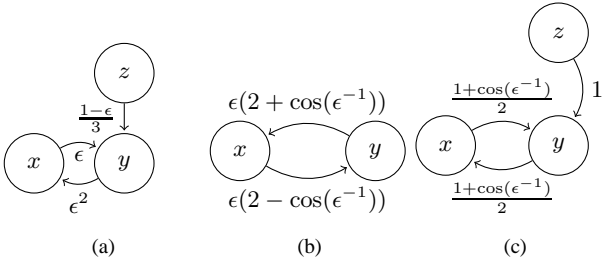

\begin{definition}[Order]\label{def:cong}
For $f,g: I\to[0,1]$, let us write $f \precsim g$ if there exist positive $b$ and $\e$ such that $f(\e') \leq b \cdot g(\e')$ for all $\e' < \e$; let $f \cong g$ stand for $f \precsim g \,\wedge\, g \precsim f$.
\end{definition}

\noindent Requiring that every two transition probability maps $f$ and $g$ occurring in the perturbation satisfy $f \precsim g$ or $g \precsim f$ rules out the example from Figure~\ref{fig:no-stable1}, but not the one from Figure~\ref{fig:no-stable2}. There $\mu_\e(z) \leq \mu_\e(x) = \frac{\e^{\cos(\e^{-1})}}{1 + \e^{\cos(\e^{-1})} (1 + \e^2)}$ and $\mu_\e(y) = \frac{1}{1 + \e^{\cos(\e^{-1})} (1 + \e^2)}$. So $\mu_\e(z) \to_{\e\to 0}0$ and $\mu_{2n\pi}(y) \to_{n\to \infty}0$ and $\mu_{2(n+1)\pi}(x) \to_{n\to \infty}0$, no state is stable. Informally, $z$ is not stable because it gives everything but receives at most $\e$; neither $x$ nor $y$ is stable since their interaction resembles Figure~\ref{fig:no-stable1} due to $\e^6$ and $\e^4 \cdot \e^{2+\cos(\e^{-1})}$. This remark is turned into a general Observation~\ref{obs:unstable-constr} below.

\begin{observation}\label{obs:unstable-constr}
For $1 \leq i \leq n$ and $1 \leq j \leq m$ let $f_i,g_j : I \to [0,1]$ be such that $\prod_i f_i$ and $\prod_j g_j$ are not $\precsim$-comparable. Then there exists a perturbation without stable states that is built only with the $f_1,\dots,f_n,g_1,\dots,g_m$ and the $1-f_1,\dots,1-f_n,1-g_1,\dots,1-g_m$. See Figure~\ref{fig:no-stable4}.
\end{observation}

\noindent Observation~\ref{obs:unstable-constr} motivates the following "unavoidable" assumption.

\begin{assumption}\label{Assum1}
The multiplicative closure of the maps $\e\mapsto p_\e(x,y)$ with $x \neq y$ is totally preordered by $\precsim$.
\end{assumption}

\noindent For example, the classical maps $\e\mapsto c\cdot \e^{\alpha}$ with $c > 0$ and $\alpha\in\mathbb{R}$ constitute a multiplicative group totally preordered by $\precsim$. One reason why we can afford such a weak Assumption~\ref {Assum1} is that we are not interested in the exact weights of some putative limit stationary distribution, but only whether the weights are bounded away from zero.

Let us show the significance of Assumption~\ref {Assum1}, which is satisfied by the perturbations in Figure~\ref{fig:stable} and \ref{fig:trans-del5}: Young's result shows that $y$ is the unique stable state of the perturbation in Figure~\ref{fig:stable1}, but it cannot say anything about Figures~\ref{fig:stable2}, \ref{fig:stable3}, and \ref{fig:trans-del5}: Figure~\ref{fig:stable2} is not regular, \textit{i.e.}, $\frac{2+\cos (\e^{-1})}{2-\cos (\e^{-1})}$ does not converge, and neither do the weights $\mu_\e(x)$ and $\mu_\e(y)$, but it is possible to show that both limits inferior are $1/4$ nonetheless, so both $x$ and $y$ are stable; the transition probabilities in Figure~\ref{fig:stable3} do not converge, and $\frac{1 + \cos(\e^{-1})}{2}$ and $1-\frac{1 + \cos(\e^{-1})}{2}$ are not even comparable, but it is easy to see that $\mu_\e(x)=\mu_\e(y) =\frac{1}{2}$; and in Figure~\ref{fig:trans-del5} $x$ is the only stable state since its weight oscillates between $\frac{1}{2}$ and $1$. Note that Assumption~\ref {Assum1} rules out the perturbations in Figure~\ref{fig:no-stable}, which have no stable state.

\section{Existence of stable states}\label{sect:pp-ess}

\noindent This section presents three transformations that simplify perturbations while retaining the relevant information about the stable states. Two of them are defined \textit{via} the dynamics of the original perturbation. The relevance of these two transformations relies on the close relation between the stationary distributions and the dynamics of Markov chains. Lemma~\ref{lem:hsr} below pinpoints this relation.

\begin{lemma}\label{lem:hsr}
A distribution $\mu$ of a finite Markov chain is stationary iff its support involves only essential states and for all states $x$ and $y$ we have $\mu(x)\p^x(\tau^+_y < \tau^+_x) = \mu(y)\p^y(\tau^+_x < \tau^+_y)$.
\end{lemma}

\noindent Lemma~\ref{lem:hsr} can already help us find the stable states of small examples such as in Figures~\ref{fig:no-stable} and \ref{fig:stable}. In Figure~\ref{fig:no-stable1} it says that $\mu_\e(x) \e^2 = \mu_\e(y) \e^{2 + \cos(\e^{-1})}$ so we find $\liminf \mu_\e(x) = \liminf \mu_\e(y) = 0$ without calculating the stationary distributions. In Figure~\ref{fig:stable2} it says that $\mu_\e(x)(2 - \cos(\e^{-1})) = \mu_\e(y)(2 + \cos(\e^{-1}))$, so $\mu_\e(x) \leq 3 \mu_\e(y)$ and $\frac{1}{4} \leq \mu_\e(y)$, and likewise for $x$.

Lemma~\ref{lem:state-del} below shows further connections between the stationary distributions and the dynamics of Markov chains. Its proof involves Lemma~\ref{lem:hsr}, and its irreducible case is used in Section~\ref{sect:td}.

\begin{lemma}\label{lem:state-del}
Let $p$ be a Markov chain with state space $S$, and let $\tilde{p}$ be defined over $\tilde{S} \subseteq S$ by $\tilde{p}(x,y) := \p^x(X_{\tau^+_{\tilde{S}}} = y)$.
\begin{enumerate}
\item\label{lem:state-del1} Then $\p^x(\tau_y < \tau^+_x) = \tilde{\p}^x(\tau_y < \tau^+_x)$ for all $x,y \in \tilde{S}$.
\item\label{lem:state-del2} Let $\mu$ ($\tilde{\mu}$) be a stationary distribution for $p$ ($\tilde{p}$). If $\tilde{S}$ are essential states, there exists  $\tilde{\mu}$ ($\mu$) a stationary distribution for $\tilde{p}$ ($p$) such that $\mu(x) = \tilde{\mu}(x) \cdot \sum_{y\in \tilde{S}}\mu(y)$ for all $x\in \tilde{S}$.
\end{enumerate}
\end{lemma}

\noindent The dynamics, \textit{i.e.}, terms like $\p^x(\tau^+_y < \tau^+_x)$ or $\p^x(X_{\tau^+} = y)$ are usually hard to compute, and so will be the two transformations that are defined \textit{via} the dynamics, but Lemma~\ref{lem:congp-congmu} below shows that approximating them is safe as far as the stable states are concerned.

\begin{lemma}\label{lem:congp-congmu}
Let $p$ and $\tilde{p}$ be perturbations with the same state space, such that $x\neq y \Rightarrow p (x,y)\cong \tilde{p}(x,y)$. For all stationary distribution maps $\mu$ for $p$, there exists $\tilde{\mu}$ for $\tilde{p}$ such that $\mu \cong \tilde{\mu}$.
\end{lemma}

\noindent \textit{E.g.}, both coefficients in Figure~\ref{fig:stable2} (\ref{fig:out-scale2}) can safely be replaced with $\e$ ($1$), and Figure~\ref{fig:ess-coll2} can be replaced with Figure~\ref{fig:ess-coll3}. Lemma~\ref{lem:congp-congmu} will dramatically simplify the computation of the stable states.

\subsection{Essential graph}\label{sect:eg}

The \emph{essential graph} of a perturbation captures the non-infinitesimal flow between different states at the normal time scale. It is a very coarse description of the perturbation.

\begin{definition}[Essential graph]\label{defn:essential-class}
Given a perturbation with state space $S$, the essential graph is a binary relation over $S$ and possesses the arc $(x,y)$ if $x \neq y$ and $p(z,t) \precsim p(x,y)$ for all $z,t\in S$. The essential classes are the sink (aka bottom) strongly connected components of the graph. The other SCCs are the transient classes. A state in an essential class is essential, the others are transient.
\end{definition}

\noindent The essential classes will be named $E_1,\dots, E_k$. Observation~\ref{obs:inf-bound} below implies that the essential graph is made of the arcs $(x,y)$ such that $x \neq y$ and $p(x,y) \cong 1$, as expected.

\begin{observation}\label{obs:inf-bound}
Let $p$ be a perturbation. There exist positive $c$ and $\e_0$  such that for all $\e < \e_0$, for all simple paths $\gamma$ in the essential graph, $c < p_\e(\gamma)$.
\end{observation}

\noindent For example, the perturbations (with $I = ]0,1]$) that are described in Figures~\ref{fig:no-stable2}, \ref{fig:no-stable3}, \ref{fig:stable1}, and \ref{fig:stable3} all have Figure~\ref{fig:ess-graph1} as essential graph, and $\{x\}$ and $\{y\}$ as essential class. Figure~\ref{fig:ess-graph2} (\ref{fig:ess-graph3}) is the essential graph of Figure~\ref{fig:ess-coll1} (\ref{fig:trans-del1}), and $\{x,y\}$ and $\{t\}$ are its essential classes. Note that the essential states of a perturbation and the essential states of a Markov chain are two distinct (yet related) concepts: \textit{e.g.}, all states from Figure~\ref{fig:ess-coll1} are essential for the Markov chain for all $\e \in ]0,1]$.

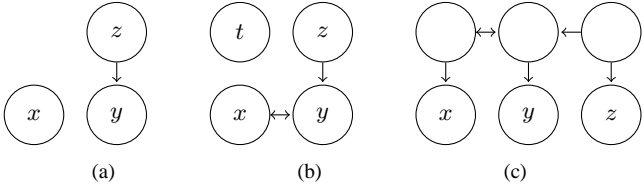
\begin{figure}
\centering
 \begin{subfigure}[b]{0.15\textwidth}
\begin{tikzpicture}[shorten >=1pt,node distance=1.1cm, auto]
  \node[state] (x) {$x$};
  \node[state] (y) [right of = x] {$y$};
    \node[state] (z) [above of = y] {$z$};
 \path[->]   (z) edge node {} (y);
 \end{tikzpicture}
 \caption{}
 \label{fig:ess-graph1}
\end{subfigure}
\begin{subfigure}[b]{0.15\textwidth}
\begin{tikzpicture}[shorten >=1pt,node distance=1.1cm, auto]
  \node[state] (x) {$x$};
  \node[state] (y) [right of = x] {$y$};
  \node[state] (t) [above of = x] {$t$};
    \node[state] (z) [above of = y] {$z$};

\path[<->] (x) edge node {} (y);
\path[->]	(z) edge node {} (y);
 \end{tikzpicture}
 \caption{}
 \label{fig:ess-graph2}
\end{subfigure}
\begin{subfigure}[b]{0.15\textwidth}
\begin{tikzpicture}[shorten >=1pt,node distance=1.1cm, auto]
  \node[state] (x) {$x$};
  \node[state] (y) [right of = x] {$y$};
  \node[state] (z) [right of = y] {$z$};
  \node[state] (x') [above of = x] {};
  \node[state] (y') [right of = x'] {};
  \node[state] (z') [right of = y'] {};
  
\path[->] (x') edge node {} (x)
	 (y') edge node {} (y)
	 (z') edge node {} (z)
	 (z') edge node {} (y');
	
\path[<->]  (x') edge node {} (y');
 \end{tikzpicture}
 \caption{}
 \label{fig:ess-graph3}
\end{subfigure}
 \caption{Essential graphs}
 \label{fig:ess-graph}
 \end{figure}

The essential graph alone cannot tell which states are stable: \textit{e.g.}, swapping $\e$ and $\e^2$ in Figure~\ref{fig:stable1} yields the same essential graph but Lemma~\ref{lem:hsr} shows that the only stable state is then $x$ instead of $y$. The graph allows us to make the following case disjunction nonetheless, along which we will either say that all states are stable, or perform one of the transformations from the next subsections.

\begin{enumerate}
\item Either the graph is empty (\textit{i.e.} totally disconnected) and the perturbation is zero, or
\item it is empty and the perturbation is non-zero, or
\item it is non-empty and has a non-singleton essential class, or
\item it is non-empty and has only singleton essential classes.
\end{enumerate}

\noindent Observation~\ref{obs:inf-bound} motivates the following convenient assumption.

\begin{assumption}\label{Assum2}
There exists $c > 0$ such that $p(\gamma) > c$ for every simple path $\gamma$ in the essential graph.
\end{assumption}

\noindent The two assumptions above do not have the same status: Assumption~\ref{Assum1} is a key condition that will appear explicitly in our final result, whereas Assumption~\ref{Assum2} is just made wlog, \textit{i.e.}, up to focusing on a smaller neighborhood of $0$ inside $I$.

Lemma~\ref{lem:ess-weight} shows the usefulness of Assumption~\ref{Assum2}. It is proved by Lemma~\ref{lem:hsr}, and is used later to strengthen Lemma~\ref{lem:state-del}.~\ref{lem:state-del2} into $\mu \cong \tilde{\mu}$.

\begin{lemma}\label{lem:ess-weight}
Let a perturbation $p$ with state space $S$ and transient states $T$ satisfy Assumption~\ref{Assum2}. Then $\frac{c}{c + |S|} \leq \sum_{x\in S\setminus T}\mu(x)$.
\end{lemma}

\subsection{Essential collapse}\label{sect:ec}


\noindent The essential collapse, defined below, amounts to merging one essential class of a perturbation into one meta-state and letting this state represent faithfully the whole class in terms of dynamics between the whole class and each of the outside states.

\begin{definition}[Essential collapse of a perturbation]\label{defn:ec}
Let $p$ be a perturbation on state space $S$. Let $x$ be a state in an essential class $E$, and let $\tilde{S} := (S\setminus E) \cup \{\cup E\}$. The essential collapse $\kappa(p,x): I\times\tilde{S}\times\tilde{S}\to[0,1]$ of $p$ around $x$ is defined below. 
\begin{align*}
\kappa(p,x)(\cup E,\cup E) & := \p^x(X_{\tau^+_{S\setminus E \cup \{x\}}} = x)\\
\kappa(p,x)(\cup E,y) & := \p^x(X_{\tau^+_{S\setminus E \cup \{x\}}} = y) & \mbox{\quad for all } y\in S\setminus E\\
\kappa(p,x)(y,\cup E) & := \sum_{z\in E}p(y,z) & \mbox{\quad for all } y\in S\setminus E\\
\kappa(p,x)(y,z) & := p(y,z) & \mbox{\quad for all } y,z\in S\setminus E
\end{align*}
\end{definition}

\begin{observation}\label{obs:ess-coll}
$\kappa(p,x)$ is again a perturbation, $\kappa$ preserves irreducibility, and if $\{x\}$ is an essential class, $\kappa(p,x) = p$.
\end{observation}

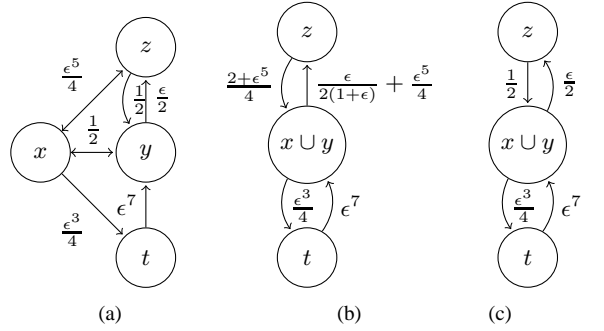
\begin{figure}
\centering
\begin{subfigure}[b]{0.15\textwidth}
\begin{tikzpicture}[shorten >=1pt,node distance=1.4cm, auto]
  \node[state] (x) {$x$};
  \node[state] (y) [right of = x] {$y$};
  \node[state] (t) [below of = y] {$t$};
    \node[state] (z) [above of = y] {$z$};

\path[<->] (x) edge node [above] {$\frac{1}{2}$} (y)		
		(x) edge node {$\frac{\e^5}{4}$} (z);
\path[->]	(z) edge [bend right] node [right] {$\frac{1}{2}$} (y)	
		(y) edge node [right] {$\frac{\e}{2}$} (z)	
		(x) edge node  [below left] {$\frac{\e^3}{4}$} (t)
		(t) edge node [left] {$\e^7$} (y);
 \end{tikzpicture}
 \caption{}
 \label{fig:ess-coll1}
\end{subfigure}
\begin{subfigure}[b]{0.2\textwidth}
\begin{tikzpicture}[shorten >=1pt,node distance=1.5cm, auto]
  \node[state] (xy) {$x \cup y$};
  \node[state] (t) [below of = xy] {$t$};
    \node[state] (z) [above of = xy] {$z$};

\path[->]	(z) edge [bend right] node [left] {$\frac{2+\e^5}{4}$} (xy)
		(xy) edge node [right] {$\frac{\e}{2(1+\e)} + \frac{\e^5}{4}$} (z)	
		(xy) edge [bend right] node  [right] {$\frac{\e^3}{4}$} (t)
		(t) edge [bend right] node [right] {$\e^7$} (xy);
 \end{tikzpicture}
 \caption{}
 \label{fig:ess-coll2}
 \end{subfigure}
\begin{subfigure}[b]{0\textwidth}
\begin{tikzpicture}[shorten >=1pt,node distance=1.5cm, auto]
  \node[state] (xy) {$x \cup y$};
  \node[state] (t) [below of = xy] {$t$};
    \node[state] (z) [above of = xy] {$z$};

\path[->]	(z) edge node [left] {$\frac{1}{2}$} (xy)
		(xy) edge [bend right] node [right] {$\frac{\e}{2}$} (z)	
		(xy) edge [bend right] node  [right] {$\frac{\e^3}{4}$} (t)
		(t) edge [bend right] node [right] {$\e^7$} (xy);
 \end{tikzpicture}
 \caption{}
 \label{fig:ess-coll3}
 \end{subfigure}
 \caption{Essential collapse}
 \label{fig:ess-coll}
 \end{figure}

\noindent For example, collapsing around $x$ or $y$ in Figure~\ref{fig:stable2} has no effect. The perturbation in Figure~\ref{fig:ess-coll1} has two essential classes, \textit{i.e.}, its essential graph has two sink SCCs, namely $\{x,y\}$ and $\{t\}$. Figure~\ref{fig:ess-coll2} displays its essential collapse around $x$.  It was calculated by noticing that $\p^x(X_{\tau^+_{\{x,z,t\}}} = t) = \frac{\e^3}{4}$, and $\p^x(X_{\tau^+_{\{x,z,t\}}} = x) = \frac{1}{2} - \frac{\e^3}{4} - \frac{\e^5}{4}+ \frac{1}{2} \cdot \p^y(X_{\tau^+_{\{x,z,t\}}} = x)$, and $\p^y(X_{\tau^+_{\{x,z,t\}}} = x) = \frac{1}{2} + \frac{1-\e}{2} \cdot \p^y(X_{\tau^+_{\{x,z,t\}}} = x)$.

Proposition~\ref{prop:essential-trans} will show that it suffices to compute the stable states of Figure~\ref{fig:ess-coll2} to compute those of Figure~\ref{fig:ess-coll1}, and by Lemma~\ref{lem:congp-congmu} it suffices to compute those of the simpler Figure~\ref{fig:ess-coll3}. However, computing the exact values $\p^x(X_{\tau^+_{S\setminus E \cup \{x\}}} = y)$ can be difficult even on simple examples like above. Fortunately, Lemma~\ref{lem:p-cong-max} shows that they are $\cong$-equivalent to maxima that are easy to compute. \textit{E.g.}, using Lemma~\ref{lem:p-cong-max} to approximate the essential collapse of Figure~\ref{fig:ess-coll1} around $x$ yields Figure~\ref{fig:ess-coll3}, but without having to compute the intermediate Figure~\ref{fig:ess-coll2}.

\begin{lemma}\label{lem:p-cong-max}
Let $p$ be a perturbation with state space $S$ satisfy Assumption~\ref{Assum1}, and let $\tilde{p}$ be the essential collapse $\kappa(p,x)$ of $p$ around $x$ in some essential class $E$. For all $y\in S\setminus E$, we have $\tilde{p}(\cup E,y) \cong \max_{z\in E} p(z,y)$ and $\tilde{p} (y,\cup E) \cong \max_{z\in E} p(y,z)$.
\end{lemma}

\noindent Note that by Lemma~\ref{lem:p-cong-max}, only the essential class is relevant during the essential collapse up to $\cong$, the exact state is irrelevant. Lemma~\ref{lem:p-cong-max} is also a tool that is used to prove, \textit{e.g.}, Proposition~\ref{prop:essential-graph} below which shows that the essential graph may contain useful information about the stable states. 

\begin{proposition}\label{prop:essential-graph}
Let a perturbation $p$ with state space $S$ satisfy Assumption~\ref{Assum1}, let $\mu$ be a corresponding stationary distribution map.
\begin{enumerate}
\item\label{prop:essential-graph1} If $y$ is a transient state, $\liminf_{\e\to 0} \mu_{\e}(y) = 0$.
\item\label{prop:essential-graph2} If two states $x$ and $y$ belong to the same essential or transient class, $\mu(x) \cong \mu(y)$.
\end{enumerate}
\end{proposition}

\noindent Proposition~\ref{prop:essential-graph}.\ref{prop:essential-graph1} says that the transient states are vanishing, \textit{e.g.} the nameless states in Figure~\ref{fig:ess-graph3}. Proposition~\ref{prop:essential-graph}.\ref{prop:essential-graph2} says that two states in the same class are either both stable or both vanishing, \textit{e.g.} $\{x\}$ and $\{y\}$ in Figure~\ref{fig:ess-graph2}.

The usefulness of the essential collapse comes from its preserving and reflecting stability, as stated in Proposition~\ref{prop:essential-trans}. Its proof invokes Lemma~\ref{lem:preserve-stable} below, which shows that the essential collapse preserves the dynamics up to $\cong$, and Lemma~\ref{lem:hsr}, which relates the dynamics and the stationary distributions.

\begin{lemma}\label{lem:preserve-stable}
Given a perturbation $p$ with state space $S$, let $\tilde{p}$ be the essential collapse of $p$ around $x$ in some essential class $E$, and let $\tilde{x} := \cup E$. The following holds for all $y\in S\setminus E$.
\[\p^y(\tau_x < \tau_y) \cong \tilde{\p}^y(\tau_{\tilde{x}} < \tau_y) \quad \wedge \quad \p^x(\tau_y < \tau_x) \cong \tilde{\p}^{\tilde{x}}(\tau_y < \tau_{\tilde{x}})\]
\end{lemma}

\begin{proposition}\label{prop:essential-trans}
Let a perturbation $p$ with state space $S$ satisfy Assumption~\ref{Assum1}, and let $x$ be in an essential class $E$.
\begin{enumerate}
\item\label{prop:essential-trans3} Let $\tilde{p}$ be the chain after the essential collapse of $p$ around $x$. Let $\mu$ ($\tilde{\mu}$) be a stationary distribution map of $p$ ($\tilde{p}$). There exists a stationary distribution map $\tilde{\mu}$ for $\tilde{p}$ ($\mu$ for $p$) such that $\tilde{\mu}(\cup E) \cong \mu(x)$ and $\tilde{\mu}(y) \cong \mu(y)$ for all $y\in S\setminus E$.
\item\label{prop:essential-trans4} A state $y\in S$ is stable for $p$ iff either $y\in E$ and $\cup E$ is stable for $\kappa(p,x)$, or $y\notin E$ and $y$ is stable for $\kappa(p,x)$.
\end{enumerate}
\end{proposition}

\noindent By definition, collapsing an essential class preserves the structure of the perturbation outside of the class, so Proposition~\ref{prop:essential-trans} implies that the essential collapse commutes up to $\cong$. Especially, the order in which the essential collapses are performed is irrelevant as far as the stable states are concerned. 

\subsection{Transient deletion}\label{sect:td}

\noindent If all the essential classes of a perturbation are singletons, Observation~\ref{obs:ess-coll} says that the essential collapse is useless. If in addition the essential graph has arcs, there are transient states, and Definition~\ref{defn:td} below deletes them to shrink the perturbation further. 


\begin{definition}[Transient deletion]\label{defn:td}
Let a perturbation $p$ with state space $S$, transient states $T$, and singleton essential classes, satisfy Assumption~\ref{Assum1}. The function $\delta(p)$ over $S\setminus T$ is derived from $p$ by transient deletion: for all distinct $x,y\in S\setminus T$ let \[\delta(p)(x,y) := \p^x(X_{\tau^+_{S\setminus T}} = y)\]
\end{definition}

\begin{observation}\label{obs:trans-del}
$\delta(p)$ is again a perturbation, $\delta$ preserves irreducibility, and if all states are essential, $\delta(p) = p$.
\end{observation}

\noindent For example, in Figure~\ref{fig:stable1} the essential classes are $\{x\}$ and $\{y\}$, $z$ is transient, and the transient deletion yields Figure~\ref{fig:trans-del4}. Also, in Figure~\ref{fig:trans-del1}, the essential classes are $\{x\}$, $\{y\}$, and $\{z\}$, the transient states are nameless, and the transient deletion yields Figure~\ref{fig:trans-del2}.

\begin{figure}
\centering
\begin{subfigure}[b]{0.25\textwidth}
\begin{tikzpicture}[shorten >=1pt,node distance=1.4cm, auto]
  \node[state] (x) {$x$};
  \node[state] (y) [right of = x] {$y$};
  \node[state] (z) [right of = y] {$z$};
  \node[state] (x') [above of = x] {};
  \node[state] (y') [right of = x'] {};
  \node[state] (z') [right of = y'] {};
  
\path[->] (x') edge node [left] {$\frac{1}{2}$} (x)
	 (y') edge node {$\frac{1}{2}$} (y)
	 (z') edge [bend left] node {$\frac{1}{2}$} (z)
	 (z') edge node [above] {$\frac{1}{2}$} (y')
	  (x) edge node [above] {$\e^2$} (y)
	  (z) edge [bend left] node [below] {$\e$} (x)	  
	  (x) edge[bend right] node [right] {$\e$} (x')
	  (z) edge [bend left] node [right] {$\e^4$} (z')
  	  (y) edge node [above] {$\e^2$} (z');

\path[<->]  (x') edge node {$\frac{1}{2}$} (y');
 \end{tikzpicture}
 \caption{}
 \label{fig:trans-del1}
\end{subfigure}
\begin{subfigure}[b]{0.2\textwidth}
\begin{tikzpicture}[shorten >=1pt,node distance=1.7cm, auto]
  \node[state] (x) {$x$};
  \node[state] (y) [right of = x] {$y$};
  \node[state] (z) [above of = y] {$z$};
  
\path[->]  (x) edge [bend left] node [below] {$\e^2 + \frac{\e}{3}$} (y)
 	  (z) edge [bend right] node [above left] {$\e + \frac{\e^4}{6}$} (x)	  
	  (y) edge [bend left] node [below] {$\frac{\e^2}{6}$} (x)
	  (z) edge node [left] {$\frac{\e^4}{3}$} (y)
  	  (y) edge [bend right] node [right] {$\frac{\e^2}{2}$} (z);
 \end{tikzpicture}
 \caption{}
 \label{fig:trans-del2}
\end{subfigure}
\begin{subfigure}[b]{0.2\textwidth}
\begin{tikzpicture}[shorten >=1pt,node distance=1.7cm, auto]
  \node[state] (x) {$x$};
  \node[state] (y) [right of = x] {$y$};
  \node[state] (z) [above of = y] {$z$};
  
\path[->]  (x) edge [bend right] node [below] {$\max(\e^2,\frac{\e}{4})$} (y)
 	  (z) edge [bend right] node [above left] {$\e$} (x)	  
	  (y) edge [bend right] node [below] {$\frac{\e^2}{8}$} (x)
	  (z) edge node [left] {$\frac{\e^4}{4}$} (y)
  	  (y) edge [bend right] node [right] {$\frac{\e^2}{2}$} (z);
 \end{tikzpicture}
 \caption{}
 \label{fig:trans-del3}
\end{subfigure}
\begin{subfigure}[b]{0.1\textwidth}
\begin{tikzpicture}[shorten >=1pt,node distance=1.5cm, auto]
  \node[state] (x) {$x$};
  \node[state] (y) [below of = x] {$y$};

\path[->] (x) edge [bend left] node {$\e$} (y)		
		(y) edge  [bend left] node  {$\e^2$} (x);
 \end{tikzpicture}
  \caption{}
  \label{fig:trans-del4}
 \end{subfigure}
\begin{subfigure}[b]{0.15\textwidth}
\begin{tikzpicture}[shorten >=1pt,node distance=2cm, auto]
  \node[state] (x) {$x$};
  \node[state] (y) [right of = x] {$y$};

\path[->] (x) edge [bend left] node {$(2^\e-1)\frac{1+\cos(\e^{-1})}{2}$} (y)		
		(y) edge  [bend left] node  {$2^\e-1$} (x);
 \end{tikzpicture}
  \caption{}
  \label{fig:trans-del5}
 \end{subfigure}

 \caption{Transient deletion (mainly)}
 \label{fig:transdel}
 \end{figure}
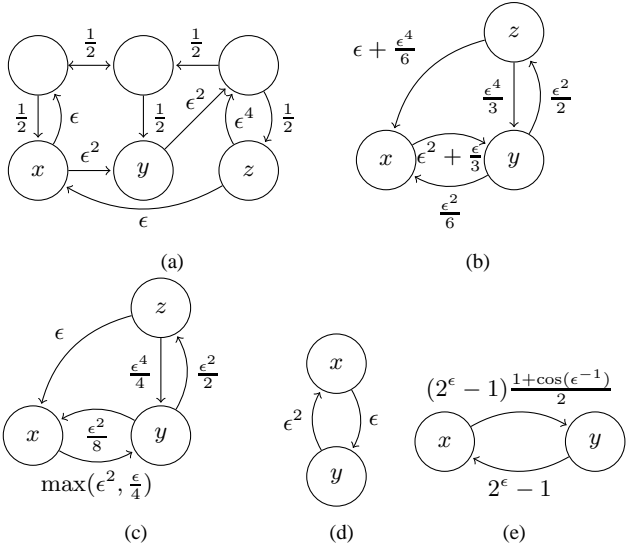

The transient deletion is useful thanks to Proposition~\ref{prop:trans-del} below, whose proof relies on Lemmas~\ref{lem:state-del}.\ref{lem:state-del2} and \ref{lem:ess-weight}.

\begin{proposition}\label{prop:trans-del}
If a perturbation $p$ satisfy Assumption~\ref {Assum1} and has singleton essential classes, $p$ and $\delta(p)$ have the same stable states.
\end{proposition}

Like the essential collapse, the transient deletion is defined \textit{via} the dynamics and is hard to compute. Like Lemma~\ref{lem:p-cong-max} did for the essential collapse, Lemma~\ref{lem:singleton-essential} approximates the transient deletion by an expression that is easy to compute.

\begin{lemma}\label{lem:singleton-essential}
If a perturbation $p$ with state space $S$ and transient states $T$ satisfies Assumption~\ref {Assum1} and has singleton essential classes,
\[\p^x(X_{\tau^+_{S\setminus T}} = y) \cong \max\{p(\gamma) : \gamma\in \Gamma_T(x,y)\} \mbox{\quad for all } x,y\in S\setminus T.\]
\end{lemma}

\noindent \textit{E.g.}, Figure~\ref{fig:trans-del1} yields Figure~\ref{fig:trans-del3} without computing Figure~\ref{fig:trans-del2}. Note that $\max(\e^2,\frac{\e}{4})$ in Figure~\ref{fig:trans-del3} may be simplified into $\e$ by Lemma~\ref{lem:congp-congmu}.

\subsection{Outgoing scaling and existence of stable states}\label{sect:os}

\noindent If the essential graph has no arc, the essential collapse and the transient deletion are useless to compute the stable states. This section says how to transform a non-zero perturbation with empty (\textit{i.e.} totally disconnected) essential graph into a perturbation with the same stable states but a non-empty essential graph, so that collapse or deletion may be applied. Roughly speaking, it is done by speeding up time until a first non-infinitesimal flow is observable between different states, \textit{i.e.} until the new essential graph has arcs.

Towards it, the \emph{ordered division} is defined in Definition~\ref{defn:div-fct}. It allows us to divide a function by a function with zeros by returning a default value in the zero case. It is named ordered because we will "divide" $f$ by $g$ only if $f \precsim g$, so that only $0$ may be "divided" by $0$. Then Observation~\ref{obs:prec-div} further justifies the terminology.

\begin{definition}[Ordered division]\label{defn:div-fct}
For $f,g: I \to [0,1]$ and $n > 1$ let us define $(f \div_n g) : I \to [0,1]$ by $(f\div_n g)(x) := \frac{f(x)}{g(x)}$ if $0 < g(x)$ and otherwise  $(f\div_n g)(x) := \frac{1}{n}$.
\end{definition}

\begin{observation}\label{obs:prec-div}
$(f \div_n g)\cdot g = f$ for all $n$ and $f,g: I \to [0,1]$ such that $f \precsim g$.
\end{observation}

\begin{definition}[Outgoing scaling]\label{defn:os}
Let a perturbation $p$ with state space $S$ satisfy Assumption~\ref{Assum1}, let $m := |S|\cdot \max\{p(z,t)\,\mid\, z,t\in S \wedge z \neq t\}$, and let us define the following.

\begin{itemize}
\item $\sigma(p)(x,y) := p(x,y) \div_{|S|} m$ for all $x \neq y$
\item $\sigma(p)(x,x) := (p(x,x) + m - 1)\div_{|S|} m$.
\end{itemize}
\end{definition}

\noindent For example, Figure~\ref{fig:stable2} satisfies Assumption~\ref{Assum1} and its essential graph is empty, \textit{i.e.} totally disconnected. Applying outgoing scaling to it yields Figure~\ref{fig:out-scale2}, which satisfies Assumption~\ref{Assum1} and whose essential graph has two arcs. Note that collapsing around $x$ or $y$ in Figure~\ref{fig:stable2} has no effect, but in Figure~\ref{fig:out-scale2} it yields a one-state perturbation. Also, Figure~\ref{fig:out-scale1} does not satisfy Assumption~\ref{Assum1} and its essential graph is empty. Applying outgoing scaling to it yields Figure~\ref{fig:out-scale3},  which does not satisfy Assumption~\ref{Assum1} and whose essential graph has one arc. Applying it again to Figure~\ref{fig:out-scale3} would only divide the non-self-loop coefficients by $3$. More generally, Proposition~\ref{prop:ogs} below states how well the outgoing scaling behaves.

\begin{figure}
\centering
\begin{subfigure}[b]{0.45\textwidth}
\begin{tikzpicture}[shorten >=1pt,node distance=2.7cm, auto]
  \node[state] (x) {$x$};
  \node[state] (y) [right of = x] {$y$};
  \node[state] (z) [right of = y] {$z$};

\path[->] (x) edge [bend left] node [above] {$\e^2 \cdot \frac{1+\cos (\e^{-1})}{2}$} (y)		
		(y) edge [bend left] node [below] {$\e^3 \cdot \frac{1+\cos (\e^{-1})}{4}$} (x)
		edge [bend left] node [above] {$\e^4 \cdot \frac{(1+\cos (\e^{-1}))^2}{4}$} (z)
		(z) edge [bend left]node  [below] {$\e^4 \cdot \frac{1+\cos (\e^{-1})}{2}$} (y);
 \end{tikzpicture}
 \caption{}
 \label{fig:out-scale1}
\end{subfigure}
\begin{subfigure}[b]{0.2\textwidth}
\begin{tikzpicture}[shorten >=1pt,node distance=1.8cm, auto]
  \node[state] (x) {$x$};
  \node[state] (y) [right of = x] {$y$};

\path[->] (x) edge [bend right] node [below] {$\frac{2-\cos(\e^{-1})}{4 + 2|\cos(\e^{-1}|}$} (y)		
		(y) edge [bend right] node [above] {$\frac{2+\cos(\e^{-1})}{4 + 2|\cos(\e^{-1}|}$} (x) ;
 \end{tikzpicture}
 \caption{}
 \label{fig:out-scale2}
\end{subfigure}
\begin{subfigure}[b]{0.2\textwidth}
\begin{tikzpicture}[shorten >=1pt,node distance=1.8cm, auto]
  \node[state] (x) {$x$};
  \node[state] (y) [right of = x] {$y$};
    \node[state] (z) [right of = y] {$z$};

\path[->] (x) edge [bend left] node [above] {$\frac{1}{3}$} (y)		
		(y) edge [bend left] node [below] {$\frac{\e}{6}$} (x)
		edge [bend left] node [above] {$\e^2 \cdot \frac{1+\cos (\e^{-1})}{6}$} (z)
		(z) edge[bend left] node  [below] {$\frac{\e^2}{3}$} (y);
 \end{tikzpicture}
 \caption{}
 \label{fig:out-scale3}
 \end{subfigure}
 \caption{Outgoing scaling}
 \label{fig:outscale}
 \end{figure}
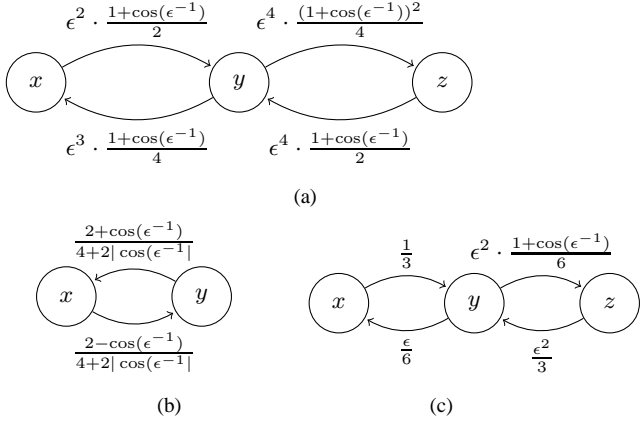

\begin{proposition}\label{prop:ogs}
\begin{enumerate}
\item\label{prop:ogs1} If a perturbation $p$ satisfies Assumption~\ref{Assum1}, so does $\sigma(p)$, and the essential graph of $\sigma(p)$ is non-empty .
\item\label{prop:ogs2} A state is stable for $p$ iff it is stable for $\sigma(p)$.
\end{enumerate}
\end{proposition}

\noindent The outgoing scaling divides the weights of the proper arcs by $m$, as if time were sped up by $m^{-1}$. The self-loops thus lose their meaning, but Proposition~\ref{prop:ogs} proves it harmless. Note that the self-loops are also ignored in Assumption~\ref{Assum1},  Lemma~\ref{lem:congp-congmu}, and Definition~\ref{defn:essential-class}.


Let us now describe a recursive procedure computing the stable states: if the perturbation is zero, all its states are stable; else, if the essential graph is empty, apply the outgoing scaling; else, apply one essential collapse or the transient deletion. This procedure is correct by Propositions~\ref{prop:ogs}.\ref{prop:ogs2}, \ref{prop:essential-trans}.\ref{prop:essential-trans4}, and \ref{prop:trans-del}, hence Theorem~\ref{thm:stable-states} below, which is the existential part of Theorem~\ref{thm:teaser}.

\begin{theorem}\label{thm:stable-states}
Let $p$ be a perturbation such that $f \precsim g$ or $g \precsim f$ for all $f$ and $g$ in the multiplicative closure of the $p(x,y)$ with $x \neq y$. Then $p$ has stable states.
\end{theorem}

\section{Abstract and quick algorithm}\label{sect:aqa}

\noindent The procedure described before Theorem~\ref{thm:stable-states} computes the stable states, but a very rough analysis of its algorithmic complexity shows that it runs in $O(n^7)$, where $n$ is the number of states. (A better analysis might find $O(n^5)$.) This bad complexity comes from the difficulty to analyze the procedure precisely and from some redundant operations done by the transformations, especially the successive essential collapses. Instead we will perform the successive collapses followed by one transient deletion as a single transformation. Applying alternately the outgoing scaling and the new transformation, both up to $\cong$, is the base of our algorithm. 

Section~\ref{sect:att} abstracts the relevant notions up to $\cong$ and gives useful algebraic properties that they satisfy. Based on these abstractions, Section~\ref{sect:algo} presents the algorithm (computing the stable states and more), its correctness, and its complexity in $O(n^3)$.

\subsection{Abstractions}\label{sect:att}

Ensuring that the essential collapse and the transient deletion can be safely performed up to $\cong$ is a straightforward sanity check, by Lemma~\ref{lem:congp-congmu}. However, the proof for the outgoing scaling involves a new algebraic structure to accommodate the ordered division, and handling the combination of the successive collapses and one deletion requires particular attention. It would have been cumbersome to define this combination directly \textit{via} the dynamics in Section~\ref{sect:pp-ess}, and more difficult to prove its correctness \textit{via} probilistic techniques, hence the usefulness of the rather atomic collapse and deletion.

Our firsts step below is to consider functions up to $\cong$.

\begin{definition}[Equivalence classes and quotient set]\label{defn:cq}
For $f: I \to [0,1]$ let $\eqclass{f}$ be its $\cong$ equivalence class; for a matrix $A = (a_{ij})_{1\leq,i,j\leq n}$ with elements in $ I \to [0,1]$, let $\eqclass{A}$ be the matrix where $\eqclass{A}_{ij} :=\eqclass{a_{ij}}$ for all $0\leq i,j\leq n$. For a set $F$ of functions from $I$ to $[0,1]$, let $\eqclass{F}$ be the quotient set $F/\cong$. Finally, it is possible to lift over $\eqclass{F}$ both $\cdot$ to $\eqclass{\cdot}$ and $\precsim$ to $\eqclass{\precsim}$.
\end{definition}

\begin{observation}\label{obs:ccc-loc}
For $(G,\cdot)$ a semigroup totally preordered by $\precsim$,
\begin{enumerate}
\item\label{obs:ccc-loc1} $\eqclass{\precsim}$ orders $\eqclass{G}$ linearly, so $\max_{\eqclass{\precsim}}$ is well-defined.
\item\label{obs:ccc-loc2} $(\eqclass{G}\cup\{\eqclass{\e\mapsto 0}\},\eqclass{\e\mapsto 0},\eqclass{\e\mapsto 1},\max_{\eqclass{\precsim}},\eqclass{\cdot})$ is a commutative semiring. (See, \textit{e.g.},~\cite{GKMS15} for the related definitions.)
\end{enumerate}
\end{observation}

\noindent The good behavior of $\cdot$ and $\precsim$ up to $\cong$ is expressed above within an existing algebraic framework, but for $\div_n$ we introduce a new algebraic structure below.

\begin{definition}[Ordered-division semiring]\label{defn:div-semiring}
An ordered-division semiring is a tuple $(F,0,1,\cdot, \leq, \div)$ such that $(F,\leq)$ is a linear order with maximum $1$, and $(F,0,1,\max_{\leq},\cdot)$ is a commutative semiring, and for all $f \leq g$ we have $f \div g$ is in $F$ and $(f \div g) \cdot g = f$.
\end{definition}

\begin{observation}\label{obs:div-semiring}
Let $(F,0,1,\cdot,\div, \leq)$ be  an ordered-division semiring. Then $0 = \min_{\leq} F$ and $f \div 1 = f$ for all $f$.
\end{observation}

\noindent Lemma~\ref{lem:odsm-f} below shows that the functions $I\to [0,1]$ up to $\cong$ form an ordered-division semiring.

\begin{lemma}\label{lem:odsm-f}
\begin{enumerate}
\item Let $n > 1$ and $f,f',g,g': I \to [0,1]$ be such that $f \cong f' \precsim g \cong g'$. Then $\eqclass{f \div_1 g} = \eqclass{f' \div_{n} g'}$, which we then write $\eqclass{f} \eqclass{\div} \eqclass{g}$.

\item For all sets $G$ of functions from $I$ to $[0,1]$ closed under multiplication, the tuple $(\eqclass{G\cup\{\e\mapsto 0\}}, \eqclass{\e\mapsto 0},\eqclass{\e\mapsto 1},\eqclass{\cdot}, \eqclass{\div}, \eqclass{\precsim})$ is an ordered-division semiring.
\end{enumerate}
\end{lemma}

\noindent For example, the set containing $\eqclass{\e\mapsto 0}$ and all the $\eqclass{\e\mapsto \e^{\alpha}}$ for non-positive $\alpha$ is an ordered-division semiring, where $\eqclass{\e\mapsto \e^{\alpha}}\eqclass{\cdot}\eqclass{\e\mapsto \e^{\beta}} = \eqclass{\e\mapsto \e^{\alpha+\beta}}$ and $\eqclass{\e\mapsto \e^{\alpha}} \eqclass{\precsim}\eqclass{\e\mapsto \e^{\beta}}$ iff $\beta \leq \alpha$, and $\eqclass{\e\mapsto 0} \eqclass{\precsim}\eqclass{\e\mapsto \e^{\alpha}}$, and $\eqclass{\e\mapsto \e^{\alpha}}\eqclass{\div}\eqclass{\e\mapsto \e^{\beta}} = \eqclass{\e\mapsto \e^{\alpha-\beta}}$ for $\beta \leq \alpha$. 

To handle $\sigma$, $\kappa$, and $\delta$ up to $\cong$ we define below transformations of weighted graphs with weights in an ordered-division semiring.

\begin{definition}[Abstract transformations]\label{defn:ao}
Let $P: S \times S \to F$, where $(F,0,1,\cdot, \leq, \div)$ is an ordered-division semiring. 
\begin{enumerate}
\item\label{defn:ao1} Let $\{(z,t)\in S^2\,|\, P(z,t) = 1 \wedge z \neq t\}$ be the essential graph of $P$, and let the sink SCCs $E_1,\dots,E_k$ be its essential classes.

\item Outgoing scaling: for $x \neq y$ let $\eqclass{\sigma}(P)(x,y) := P(x,y) \div M$, where $M := \max_{\leq}\{P(z,t) : (z,t)\in S\times S \wedge z\neq t\}$, and $\eqclass{\sigma}(P)(x,x) := 1$.

\item Essential collapse: let $\eqclass{\kappa}(P,E_i)$ be the matrix with states $\{\cup E_i\}\cup S\setminus E_i$ such that for all $x,y\in S\setminus E_i$ we set $\eqclass{\kappa}(P, E_i)(x,y) := P(x,y)$ and $\eqclass{\kappa}(P, E_i)(\cup E_i,y) := \max_{\leq}\{P(x_i,y) : x_i\in E_i\}$ and $\eqclass{\kappa}(P, E_i)(x,\cup E_i) := \max_{\leq}\{P(x,x_i) : x_i\in E_i\}$ and $\eqclass{\kappa}(P, E_i)(\cup E_i,\cup E_i) := 1$.

\item Shrinking: let $\eqclass{\chi}(P)$ be the matrix with state space $\{\cup E_1,\dots,\cup E_k\}$ such that  for all $i,j$ $\eqclass{\chi}(P)(\cup E_i,\cup E_j) := \max_{\leq}\{P(\gamma) : \gamma\in \Gamma_T(E_i,E_j)\}$.
\end{enumerate}
\end{definition}

\noindent In Definition~\ref{defn:ao}, the weights $P(x,x)$ occur only in the definitions of the self-loops of the transformed graphs, whence Observation~\ref{obs:ao-refl-irrel} below.

\begin{observation}\label{obs:ao-refl-irrel}
Let $(F,0,1,\cdot, \leq, \div)$ be an ordered-division semiring, let $P,P': S \times S \to F$ be such that $P(x,y) = P'(x,y)$ for all $x \neq y$. Then $P$ and $P'$ have the same essential graph and classes $E_1,\dots,E_k$; $\eqclass{\sigma}(P)(x,y) = \eqclass{\sigma}(P')(x,y)$ for all $x \neq y$; and for all $l$ and $i \neq j$ we have $\eqclass{\chi}(P)(\cup E_i,\cup E_j) = \eqclass{\chi}(P')(\cup E_i,\cup E_j)$ and $\eqclass{\kappa}(P)(P,E_l)(\cup E_i,\cup E_j) = \eqclass{\kappa}(P')(P',E_l)(\cup E_i,\cup E_j)$.
\end{observation}

\noindent Lemma~\ref{lem:pop} below shows that the transformations from Definition~\ref{defn:ao} are faithful abstractions of $\sigma$, $\kappa$, and $\delta$. Some proofs come with examples, which also highlight the benefits of abstraction.

\begin{lemma}[Abstract and concrete transformations]\label{lem:pop}
Let a perturbation $p$ with state space $S$ satisfy Assumption~\ref{Assum1}, let $E_1,\dots, E_k$ be its essential classes, and for all $i$ let $x_i \in E_i$.
\begin{enumerate}
\item\label{lem:pop0} $p$ and $\eqclass{p}$ have the same essential graph.
\item\label{lem:pop1}  $\eqclass{\sigma}(\eqclass{p})(x,y) = \eqclass{\sigma(p)}(x,y)$ for all $x \neq y$.
\item\label{lem:pop2} $\eqclass{\chi}(\eqclass{p})(\{x\},\{y\}) = \eqclass{\delta(p)}(x,y)$ whenever $\delta(p)$ is well-defined.
\item\label{lem:pop3} $\eqclass{\kappa}(\eqclass{p},E_i) = \eqclass{\kappa(p,x_i)}$.
\item\label{lem:pop4} $\eqclass{\chi}(\eqclass{p}) = \eqclass{\chi}\circ\eqclass{\kappa}(\eqclass{p},E_1)$.
\item\label{lem:pop5} $\eqclass{\delta\circ\kappa(\dots\kappa(\kappa(\sigma(p),x_1),x_2)\dots,x_{k})}(\cup E_i,\cup E_j) = \eqclass{\chi}\circ\eqclass{\sigma}(\eqclass{p})(\cup E_i,\cup E_j)$ for all $i \neq j$.
\end{enumerate}
\end{lemma}

\noindent By the algorithm underlying Theorem~\ref{thm:stable-states} and Lemma~\ref{lem:pop}.\ref{lem:pop5} we are now able to state the following.

\begin{proposition}\label{prop:abs-algo}
Let a perturbation $p$ satisfy Assumption~\ref{Assum1}. There exists $n\in\mathbb{N}$ such that $(\eqclass{\chi}\circ\eqclass{\sigma})^n(\eqclass{p})(x,y) = 0$ for all $x \neq y$ in its state space. Furthermore, the states of $(\eqclass{\chi}\circ\eqclass{\sigma})^n(\eqclass{p})$ correspond to the stable states of $p$.
\end{proposition}

\subsection{The algorithm}\label{sect:algo}

\noindent Algorithm~\ref{algo:br} mainly consists in applying recursively the function $\eqclass{\chi}\circ\eqclass{\sigma}$ occurring in Proposition~\ref{prop:abs-algo} until an empty (\textit{i.e.} totally disconnected) graph is produced. It does not explicitly refer to perturbations since this notion was abstracted on purpose. Instead, the algorithm manipulates digraphs with arcs labeled in an ordered-division semiring, in which inequality, multiplication and ordered division are implicitly assumed to be computable.

\begin{algorithm}
\SetKwProg{Fn}{Function}{ is}{end}
\SetKwFunction{KwHub}{Hub}
\SetKwFunction{KwHubRec}{HubRec}
\SetKwFunction{KwTarjanSinkSCC}{TarjanSinkSCC}
\SetKwFunction{KwDijkstraSource}{Dijkstra}
\SetKwFunction{KwDijkstraSink}{DijkstraSink}

\Fn{\KwHub}{
\SetKwInOut{Input}{input}\SetKwInOut{Output}{output}
\Input{$(S,P)$, where $P: S\times S \to F$}
 \tcp*[h]{$(F,0,1,\cdot, \leq,\div)$ is an ordered-division semiring.}
 
\Output{a subset of $S$}

\BlankLine

$\hat{S} \leftarrow \{\{s\}|s\in S\}$\tcp*[r]{For bookkeeping.}\label{line:type-vertex}
\lFor(\tcp*[f]{For bookkeeping.}\label{line:type-edge}){$x,y\in S$}{$\hat{P}(\{x\},\{y\}) \leftarrow P (x,y)$}
\KwRet{\KwHubRec{$\hat{S}$,$\hat{P}$}}\;
}

\BlankLine

\Fn{\KwHubRec}{
\SetKwInOut{Input}{input}\SetKwInOut{Output}{output}
\Input{$(S,P)$, where $S$ is a set of sets and $P:S\times S\to F$}
\Output{a subset of $S$}

\BlankLine

$M \leftarrow \max\{P(x,y)\,|\,(x,y)\in S\times S\wedge x\neq y\}$\;\label{line:max-label}
\lIf(\tcp*[f]{Recursion base case}){$M = 0$}{\KwRet{$\cup S$}}\label{line:end-rec}
\lFor(\tcp*[f]{Outgoing scaling.}){$x,y\in S$ and $x \neq y$}{$P(x,y)\leftarrow P(x,y) \div M$}\label{line:norm-label}

\BlankLine

$A \leftarrow \{(x,y)\in S\times S\,|\,P(x,y) = 1\}$\tcp*[r]{$A$ is a digraph.}\label{line:max-label-edges}

$(E_1,\dots,E_k) \leftarrow$\KwTarjanSinkSCC{$S$,$A$}\tcp*[r]{Returns the sink SCCs of $A$.}\label{line:Tarjan-SSCC}

\BlankLine

\tcp*[h]{Maximal labels of direct arcs, below.}

\lFor(\label{line:partial-max}){$1\leq i,j \leq k$}{$P'(\cup E_i,\cup E_j) \leftarrow \max\{P(x,y)\,|\,(x,y)\in E_i\times E_j\}$}

\BlankLine

\tcp*[h]{Maximal labels of all relevant paths, in the remainder.}

$T \leftarrow S\setminus (E_1\cup\dots\cup E_k)$\;\label{line:transient}
$P_T \leftarrow P$\tcp*[r]{Initialisation.}\label{line:reduce-transient}

\lFor(\tcp*[f]{Drops arcs not starting in $T$.\label{line:remove}}){$(x,y)\in (S\setminus T)\times S$}{$P_T(x,y) \leftarrow 0$}
\lFor(\label{line:Dijkstra}){$y\in T$}{$P_T(y,\_)\leftarrow$\KwDijkstraSource{$S$,$P_T$,$y$, $\cdot$, $\max$}}
\tcp*[h]{$P_T(y,\_)$ is the "distance" function from $y\in T$, using $\cdot$ and $\max$.}

\For{$1\leq i,j \leq k$ and $i\neq j$ and $(x_i,x_j,y)\in E_i\times E_j\times T$}{$P'(\cup E_i,\cup E_j) \leftarrow \max(P'(\cup E_i,\cup E_j),P(x_i,y) \cdot P_T(y,x_j))$\;\label{line:global-max}}

\KwRet{\KwHubRec{$\{\cup E_1,\dots,\cup E_k\},P'$}}\label{line:return}
}
\caption{Hub}\label{algo:br}
\end{algorithm}

One call to the recursive function \KwHubRec corresponds to $\eqclass{\chi}\circ\eqclass{\sigma}$, \textit{i.e.} Lines~\ref{line:max-label} and \ref{line:norm-label} correspond to $\eqclass{\sigma}$, and Lines~\ref{line:max-label-edges} till \ref{line:global-max} correspond to $\eqclass{\chi}$. Before calling \KwHubRec Lines~\ref{line:type-vertex} and \ref{line:type-edge} produce an isomorphic copy of the input, which will be easier to handle when making unions and keeping track of the remaining vertices. Note that Line~\ref{line:norm-label} does not update the $P(x,x)$. It would be useless indeed, since in Definition~\ref{defn:ao} the self-loops of the original graph occur only in the definition of the self-loops of the transformed graphs, and since self-loops are irrelevant by Obsevation~\ref{obs:ao-refl-irrel}. Line~\ref{line:max-label-edges} computes the essential graph, up to self-loops, and Line~\ref{line:Tarjan-SSCC} computes the essential classes by a modified Tarjan's algorithm, as detailed in Algorithm~\ref{algo:MT}. The computation of $\eqclass{\chi}(P)(\cup E_i,\cup E_j) := \max_{\leq}\{P(\gamma) : \gamma\in \Gamma_T(E_i,E_j)\}$ is performed in two stages: the first stage at Line~\ref{line:partial-max} considers only paths of length one; the second stage at Line~\ref{line:global-max} considers only paths of length greater than one, and therefore having their second vertex in $T$. This case disjunction reduces the size of the graph on which the shortest path algorithm from Line~\ref{line:Dijkstra} is run, and thus reduces the complexity of \KwHub from $O(n^4)$ to $O(n^3)$, as will be detailed in Proposition~\ref{prop:complexity}. Note that the shortest path algorithm is called with laws $\cdot$ and $\max$ instead of $+$ and $\min$. Moreover, since our weights are at most $1$ we can use \cite{LGJLMPS57} or \cite{Dijkstra59} (which assume non-negative weights) to implement Line~\ref{line:Dijkstra}. 

Proposition~\ref{prop:complexity} below shows that our algorithm is fast.

\begin{proposition}\label{prop:complexity}
The algorithm \KwHub terminates within $O(n^3)$ computation steps, where $n$ is the number of vertices of the input graph.
\end{proposition}

\noindent By Propositions~\ref{prop:abs-algo} and \ref{prop:complexity} we now state our main algorithmic result, which is the algorithmic part of Theorem~\ref{thm:teaser}.

\begin{theorem}\label{thm:vanish-time}
Let a perturbation $p$ satisfy Assumption~\ref {Assum1}. A state is stochastically stable iff it belongs to $\KwHub(S,\eqclass{p})$. Provided that inequality, multiplication, and ordered division between equivalence classes of perturbation maps can be computed in constant time, stability can be decided in $O(n^3)$, where $n$ is the number of states. 
\end{theorem}

\noindent One of the achievements of our algorithm is that it processes all weighted digraphs (\textit{i.e.} abstractions of perturbations) uniformly. Especially, neither irreducibility nor any kind of connectedness is required. For example in Figure~\ref{fig:disc-pert}, the four-state perturbation is the disjoint union of two smaller perturbations. As expected the stable states of the union are the union of the stable states, \textit{i.e.} $\{x,z\}$, but whereas the outgoing scaling applied to the bottom of Figure~\ref{fig:disc-pert2} (the perturbation restricted to $\{z,t\}$) would yield the bottom of Figure~\ref{fig:disc-pert5} directly by division by $\eqclass{\e^6}$, two rounds of ougtoing scaling lead to this stage when processing the four-state perturbations.

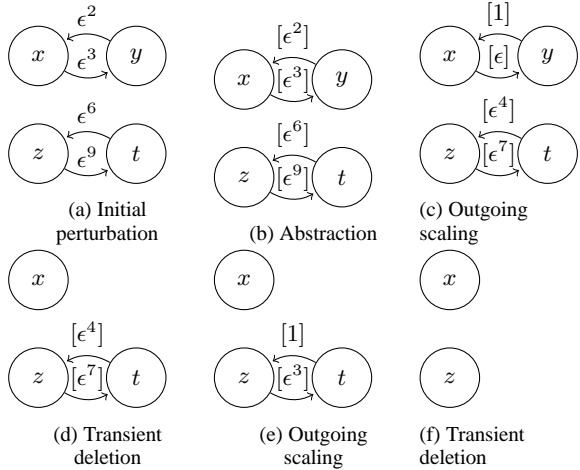
\begin{figure}
\centering
\begin{subfigure}[b]{0.15\textwidth}
\begin{tikzpicture}[shorten >=1pt,node distance=1.3cm, auto]
  \node[state] (x) {$x$};
  \node[state] (y) [right of = x] {$y$};
  \node[state] (z) [below of = x] {$z$};
  \node[state] (t) [right of = z] {$t$};
  
\path[->] (x) edge [bend right] node  {$\e^3$} (y)		
		(y) edge [bend right] node [above]{$\e^2$} (x)
		(z) edge [bend right] node  {$\e^9$} (t)
		(t) edge [bend right] node [above]{$\e^6$} (z);
 \end{tikzpicture}
 \caption{Initial\\perturbation}
 \label{fig:disc-pert1}
\end{subfigure}
\begin{subfigure}[b]{0.15\textwidth}
\begin{tikzpicture}[shorten >=1pt,node distance=1.3cm, auto]
  \node[state] (x) {$x$};
  \node[state] (y) [right of = x] {$y$};
  \node[state] (z) [below of = x] {$z$};
  \node[state] (t) [right of = z] {$t$};
  
\path[->] (x) edge [bend right] node  {$\eqclass{\e^3}$} (y)		
		(y) edge [bend right] node [above]{$\eqclass{\e^2}$} (x)
		(z) edge [bend right] node  {$\eqclass{\e^9}$} (t)
		(t) edge [bend right] node [above]{$\eqclass{\e^6}$} (z);
 \end{tikzpicture}
 \caption{Abstraction}
 \label{fig:disc-pert2}
\end{subfigure}
\begin{subfigure}[b]{0.1\textwidth}
\begin{tikzpicture}[shorten >=1pt,node distance=1.3cm, auto]
  \node[state] (x) {$x$};
  \node[state] (y) [right of = x] {$y$};
  \node[state] (z) [below of = x] {$z$};
  \node[state] (t) [right of = z] {$t$};
  
\path[->] (x) edge [bend right] node  {$\eqclass{\e}$} (y)		
		(y) edge [bend right] node [above]{$\eqclass{1}$} (x)
		(z) edge [bend right] node  {$\eqclass{\e^7}$} (t)
		(t) edge [bend right] node [above]{$\eqclass{\e^4}$} (z);
 \end{tikzpicture}
 \caption{Outgoing\\scaling}
 \label{fig:disc-pert3}
\end{subfigure}
\begin{subfigure}[b]{0.15\textwidth}
\begin{tikzpicture}[shorten >=1pt,node distance=1.3cm, auto]
  \node[state] (x) {$x$};
  \node[state] (z) [below of = x] {$z$};
  \node[state] (t) [right of = z] {$t$};
  
\path[->] 	(z) edge [bend right] node  {$\eqclass{\e^7}$} (t)
		(t) edge [bend right] node [above]{$\eqclass{\e^4}$} (z);
 \end{tikzpicture}
 \caption{Transient\\deletion}
 \label{fig:disc-pert4}
\end{subfigure}
\begin{subfigure}[b]{0.15\textwidth}
\begin{tikzpicture}[shorten >=1pt,node distance=1.3cm, auto]
  \node[state] (x) {$x$};
  \node[state] (z) [below of = x] {$z$};
  \node[state] (t) [right of = z] {$t$};
  
\path[->] 	(z) edge [bend right] node  {$\eqclass{\e^3}$} (t)
		(t) edge [bend right] node [above]{$\eqclass{1}$} (z);
 \end{tikzpicture}
 \caption{Outgoing\\scaling}
 \label{fig:disc-pert5}
\end{subfigure}
\begin{subfigure}[b]{0.1\textwidth}
\begin{tikzpicture}[shorten >=1pt,node distance=1.3cm, auto]
  \node[state] (x) {$x$};
  \node[state] (z) [below of = x] {$z$};

 \end{tikzpicture}
 \caption{Transient\\deletion}
 \label{fig:disc-pert6}
\end{subfigure}
 \caption{The algorithm run on a disconnected perturbation}
 \label{fig:disc-pert}
 \end{figure}

\section{Discussion}\label{sect:disc}

This section studies two special cases of our setting: first, how assumptions that are stronger than Assumption~\ref{Assum1} make not only some proofs easier but also one result stronger; second, how far Young's technique can be generalized. Then we notice that the termination of our algorithm defines an induction proof principle, which is used to show that the algorithm computes a well-known object when fed a strongly connected graph. Eventually, we discuss how to give the so-far-informal notion of time scale a formal flavor.

\subsection{Stronger assumption}

\noindent Let us consider Assumption~\ref{Assum3}, which is stronger version of Assumption~\ref{Assum1}. Assumption~\ref{Assum3} yields Proposition~\ref{prop:transient-vanish}, which is stronger version of Proposition~\ref{prop:essential-graph}.\ref{prop:essential-graph1}. (The proofs are similar but the new one is simpler.)

\begin{assumption}\label{Assum3}
If $x\neq y$ and $p(x,y)$ is non-zero, it is positive; and $f \cong g$ or $f\in o(g)$ or $g\in o(f)$ for all $f$ and $g$ in the multiplicative closure of the $\e\mapsto p_\e(x,y)$ with $x \neq y$.
\end{assumption}

\noindent 

\begin{proposition}\label{prop:transient-vanish}
Let a perturbation $p$ with state space $S$ satisfy Assumption~\ref{Assum3}, and let $\mu$ be a stationary distribution map for $p$. If $y$ is a transient state, $\lim_{\e\to 0} \mu_{\e}(y) = 0$.
\end{proposition}

\noindent Under Assumption~\ref{Assum1} some states may be neither stable nor fully vanishing: $y$ in Figure~\ref{fig:trans-del5} and $x$ in Figure~\ref{fig:no-stable1} where the bottom $\e^2$ is replaced with $\e$. Assumption~\ref{Assum3} rules out such cases, as below.

\begin{corollary}\label{cor:strong-assum-stable}
If a perturbation $p$ satisfies Assumption~\ref{Assum3}, every state is either stable or fully vanishing. 
\end{corollary}

\subsection{Generalization of Young's technique}

\noindent Our approach to prove the existence of and compute the stable states of a perturbation is different from Young's approach~\cite{Young93} which uses a finite version of the Markov chain tree theorem. In this section we investigate how far Young's technique can be generalized. This will suggest that we were right to change approaches, but it will also yield a decidability result in Proposition~\ref{prop:stable-dec}.

Lemma~\ref{lem:gen-Young} below is a generalization of \cite[Lemma 1]{Young93}. Both proofs use the Markov chain tree theorem, but they are significantly different nonetheless. Let $p$ be a perturbation with state space $S$. As in~\cite{Young93} or \cite{GKMS15}, for all $x\in S$ let $\mathcal{T}_x$ be the set of the spanning trees of (the complete graph of) $S \times S$ that are directed towards $x$. For all $x\in S$ let $\beta^x_\e := \max_{T\in \mathcal{T}_x} \prod_{(z,t)\in T}p_\e(z,t)$.

\begin{lemma}\label{lem:gen-Young}
A state $x$ of an irreducible perturbation with state space $S$ is stable iff $\beta^{y} \precsim \beta^{x}$ for all $y\in S$.
\end{lemma}

\noindent Assumption~\ref{Assum1} and Lemma~\ref{lem:gen-Young} together yield Observation~\ref{obs:gen-pos-ss}, a generalization of existing results about existence of stable states, such as \cite[Theorem 4]{Young93}. The underlying algorithm runs in time $O(n^3)$ where $n$ is the number of states, just like Young's.

\begin{observation}\label{obs:gen-pos-ss}
Let a perturbation $p$ on state space $S$ satisfy Assumption~\ref{Assum1}. If for all $x \neq y$ the map $p(x,y)$ is either identically zero or strictly positive, $p$ has stable states.
\end{observation}

\noindent The stable states of a perturbation are computable even without the positivity assumption from Observation~\ref{obs:gen-pos-ss}, but their existence is no longer guaranteed by the same proof. In this way, Observation~\ref{obs:gen-ss} is like the existential part of Theorem~\ref{thm:teaser}, but with a bad complexity.

\begin{observation}\label{obs:gen-ss}
Let $F$ be a set of perturbation maps of type $I\to [0,1]$ for some $I$. Let us assume that $F$ is closed under multiplication by elements in $F$ and by characteristic functions of decidable subsets of $I$, that $\precsim$ is decidable on $F\times F$, and that the supports of the functions in $F$ are uniformly decidable. If $f\precsim g$ or $g\precsim f$ for all $f,g\in F$, stability is decidable  in $O(n^5)$ for the perturbations $p$ such that $x \neq y \Rightarrow p(x,y) \in F$.
\end{observation}

\noindent The assumption $f\precsim g$ or $g\precsim f$ for all $f,g\in F$ from Observation~\ref{obs:gen-ss} corresponds to Assumption~\ref{Assum1}. Proposition~\ref{prop:stable-dec} below drops it while preserving decidability of stability, but at the cost of an exponential blow-up because the supports of the maps are no longer ordered by inclusion.

\begin{proposition}\label{prop:stable-dec}
Let $F$ be a set of perturbation maps of type $I\to [0,1]$ for some $I$. Let us assume that $F$ is closed under multiplication by elements in $F$ and by characteristic functions of decidable subsets of $I$, that $\precsim$ is decidable on $F\times F$, and that the supports of the functions in $F$ are uniformly decidable. Then stability is decidable for the perturbations $p$ such that $x \neq y \Rightarrow p(x,y) \in F$.
\end{proposition}

\subsection{What does Algorithm~\ref{algo:br} compute?}

Applying sequentially outgoing scaling, essential collapse, and transient deletion terminates. So it amounts to an \emph{induction proof principle} for finite graphs with arcs labeled in an ordered-division semiring. Observation~\ref{obs:span-tree} is proved along this principle. It can also be proved by a very indirect argument using Lemma~\ref{lem:gen-Young} and Theorem~\ref{thm:vanish-time}, but the proof using induction is simple and from scratch.

\begin{observation}\label{obs:span-tree}
Let $(F,0,1,\cdot, \leq,\div)$ be an ordered-division semiring, and let $P:S\times S \to F$ correspond to a strongly connected digraph, where an arc is absent iff its weight is $0$. Then $\KwHub(S,P)$ returns the roots of the maximum directed spanning trees.
\end{observation}

\noindent Note that finding the roots from Observation~\ref{obs:span-tree} is also doable in $O(n^3)$ by computing the maximum spanning trees rooted at each vertex, by~\cite{GGST86} which uses the notion of \emph{heap}, whereas $\KwHub$ uses a less advanced algorithm.

Observation~\ref{obs:span-tree} may be extended to non strongly connected digraphs by considering the sink SCCs independently, but alternatively it is not obvious how to generalize the notion of maximal spanning tree into a notion that is meaningful for non-strongly connected graphs. Nevertheless, the vertices returned by $\KwHub(S,P)$ are the one in $S$ that somehow attract the more flow/traffic according to $P$, hence the name $\KwHub$.

One last algorithmic remark: from the proof of Proposition~\ref{prop:complexity} we see that Tarjan's algorithm is an overkill to get a complexity of $O(n^3)$. Indeed, combining several basic shortest path-algorithms would have done the trick, but using Tarjan's algorithm should make the computation of $\KwHub$ faster by a constant factor.

\subsection{Vanishing time scales}

Under Assumption~\ref{Assum1}, computing $\KwHub$ and considering the intermediate weighted graphs shows the order in which the states are found to be vanishing. Under the stronger Assumption~\ref{Assum3}, a notion of \emph{vanishing time scale} may be defined, with the flavor of non-standard analysis~\cite{Robinson74}. Let $(\mathcal{T},\cdot)$ be a group of functions $I \to ]0,+\infty[$ such that $f \cong g$ or $f\in o(g)$ or $g\in o(f)$ for all $f$ and $g$ in $\mathcal{T}$. The elements of $\eqclass{\mathcal{T}}$ are called the time scales. Let a perturbation $p$ on state space $S$ satisfy Assumption~\ref{Assum3} and let $x \in S$ be deleted at the $d$-th recursive call of $\KwHub(S,\eqclass{p})$. Let $M_1,\dots,M_d$ be the maxima (\textit{i.e.} $M$) from Line~\ref{line:max-label} in Algorithm~\ref{algo:br} at the 1st,...,$d$-th recursive calls, respectively. We say that $x$ vanishes at time scale $\prod_{1 \leq i \leq d} M_i^{-1}$.

Figure~\ref{fig:vanish-ts} suggests that a similar account of vanishing times scales, even just a qualitative one, would be much more difficult to obtain by invoking the Markov chain tree theorem as in~\cite{Young93}. The only stable states is $t$; the state $z$ vanishes at time scale $\eqclass{\e}^{-2}$; and $x$ and $y$ vanish at the same time scale $\eqclass{1}$ although the maximum spanning trees rooteed at $x$ and $y$ have different weights: $\e^4$ and $\e^3$, respectively.

\begin{figure}
\centering
\begin{tikzpicture}[shorten >=1pt,node distance=1.5cm, auto]
 \node[state] (z) {$z$};
  \node[state] (x) [right of = z] {$x$};
  \node[state] (y) [right of = x] {$y$};
   \node[state] (t) [right of = y] {$t$};

\path[->] (z) edge [loop left] node {$1-\e$} ()
		edge  [bend left]node {$\e$} (x)
		(x) edge [bend left]node {$1-\e$} (z)
		edge  [bend left] node {$\e$} (y)
		(y) edge  [bend left] node {$1-\e^2$} (t)
		edge [bend left] node {$\e^2$} (x)
		(t) edge [loop right] node {$1-\e$} ()
		edge  [bend left] node {$\e$} (y);
 \end{tikzpicture}
 \caption{Vanishing time scale}
 \label{fig:vanish-ts}
 \end{figure}
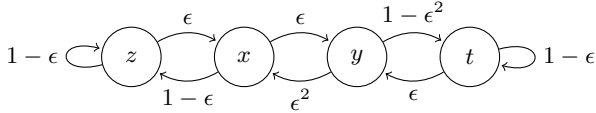

\acks

We thank Ocan Sankur for useful discussions.


\bibliographystyle{abbrvnat}
\bibliography{article}

\appendix

\section{Tarjan Modified}

\noindent The function \emph{TarjanSinkSCC} is written in Algorithm~\ref{algo:MT}. It consists of Tarjan's algorithm~\cite{Tarjan72},\cite{wiki:TarjanSCC}, which normally returns all the SCCs of a directed graph, plus a few newly added lines (as mentioned in comments) so that it returns the sink SCCs only. It is not difficult to see that the newly added lines do not change the complexity of the algorithm, which is $O(|S|+|A|)$ where $|S|$ and $|A|$ are the numbers of vertices and arcs in the graph, respectively. The new lines only deal with the new boolean values $v.sink$. These lines are designed so that when popping an SCC with root $v$ from the stack , the value $v.sink$ is true iff the SCC is a sink, hence the test at Line~\ref{line:actual-sink}. All the $v.sink$ are initialized with $true$ at Line~\ref{line:apriori-sink}, and $v.sink$ is set to false at two occasions: at Line~\ref{line:above-new-sink} before a sink SCC with root $v$ is output; and at Line~\ref{line:above-old-SCC} when one successor $w$ of $v$ has already been popped from the stack (since $w.index$ is defined), which means that there is one SCC below that of $v$. These are then propagated upwards at Line~\ref{line:no-sink-up}. The conjunction reflects the facts that a vertex is not in a sink SCC iff one of its successors in the same SCC is not.  

\begin{algorithm}
\SetKwProg{Fn}{Function}{ is}{end}
\SetKwFunction{KwStrongConnect}{StrongConnect}
\Fn{\KwTarjanSinkSCC}{
\SetKwInOut{Input}{input}\SetKwInOut{Output}{output}
\Input{$(S,A)$, where $A\subseteq S\times S$}
\Output{the sink SCC}

\BlankLine

$index \leftarrow 0$\;
$stack \leftarrow \emptyset$\;

\lFor(){$v\in S$}{$v.onstack \leftarrow false$}\label{line:not-on-stack}

\lFor(\tcp*[f]{Newly added.}){$v\in S$}{$v.sink \leftarrow true$}\label{line:apriori-sink}

\For{$v\in V$}{
\lIf(){$v.index$ is undefined}{\KwStrongConnect{$v$}}
}

\BlankLine

\Fn{\KwStrongConnect{$v$}}{
$v.index \leftarrow index$\;
$v.lowlink \leftarrow index$\;
$index \leftarrow index + 1$\;
$stack.push(v)$\;
$v.onstack \leftarrow true$\;

\BlankLine

\For{$(v,w)\in A$}{
\eIf{$w.index$ is undefined}{
\KwStrongConnect{$w$}\;
$v.lowlink \leftarrow \min(v.lowlink,w.lowlink)$\;
$v.sink \leftarrow v.sink\,\wedge\,w.sink$\tcp*[r]{Newly added.}\label{line:no-sink-up}
}{
\eIf{$w.onstack = true$}{
$v.lowlink \leftarrow \min(v.lowlink,w.index)$
}{
$v.sink \leftarrow false$\tcp*[r]{Newly added.}\label{line:above-old-SCC}
}
}
}
\BlankLine

\If{$v.lowlink = v.index$}{
start a new SCC\;

\Repeat{w = v}{
$w \leftarrow stack.pop()$\;
$w.onstack \leftarrow false$\;
add $w$ to the current SCC\;
}

\If(\tcp*[f]{Newly added.}){$v.sink$\label{line:actual-sink}}{
$v.sink \leftarrow false$\tcp*[r]{Newly added.}\label{line:above-new-sink}
output the SCC\;
}
}

}
}
\caption{modification of Tarjan's SCC algorithm\label{algo:MT}}
\end{algorithm}

\section{Proofs and Lemmas}

Lemma~\ref{lem:cong} below relates to Definition~\ref{def:cong}.

\begin{lemma}\label{lem:cong}
\begin{enumerate}
\item\label{lem:cong1} $\precsim$ is a preorder and $\cong$ an equivalence relation.
\item\label{lem:cong2} For all $f,g: I\to ]0,1]$, we have $f \precsim g$ iff $\frac{1}{g} \precsim \frac{1}{f}$, so $f \cong g$ iff $\frac{1}{f} \cong \frac{1}{g}$.
\item\label{lem:cong3} $f \precsim g$ and $f' \precsim g'$ implies $f+f' \precsim g+g'$ and $f\cdot f' \precsim g\cdot g'$.
\item\label{lem:cong4} $f \cong g$ and $f' \cong g'$ and $f \precsim f'$ implies $g \precsim g'$.
\item\label{lem:cong5} $f+f' \cong \max(f,f') := x \mapsto \max(f(x),f'(x))$.
\item\label{lem:cong6} $f \precsim f'$ implies $\max(f,f') \cong f'$.
\item\label{lem:cong7} $f\mid_J \precsim g\mid_J$ and $f\mid_{I\setminus J} \precsim g\mid_{I\setminus J}$ implies $f \precsim g$.
\item\label{lem:cong8} Let $0$ be a limit point of both $J \subseteq I$ and $I \setminus J$. A state $x$ is stable (fully vanishing) for a perturbation $p$ iff it is stable (fully vanishing) for both $p\mid_J$ and $p\mid_{I\setminus J}$.
\end{enumerate}
\end{lemma}

\begin{proof}[Observation~\ref{obs:unstable-constr}]
Let $S := \{x_1,\dots,x_n, y_1,\dots,y_m\}$, for all $i < n-1$ let $p(x_i,x_{i+1}) := f_i$, let $p(x_n,y_0) := f_n$, for all $i$ let $p(x_i,x_1) := 1 - f_i$, for all $j < m-1$ let $p(y_j,y_{j+1}) := g_j$, let $p(y_m,x_0) := g_m$, for all $j$ let $p(y_j,y_1) := 1 - g_j$. It is easy to check that $\beta^{x_i} = \prod_j g_j \cdot \prod_{1 \leq k < i} f_i \cdot \prod_{i \leq k} \max(f_k,1-f_k) \cong \prod_j g_j \cdot \prod_{1 \leq k < i} f_i $ for all $i$. Let us first assume that the functions in $F$ are positive  So by Lemma~\ref{lem:gen-Young}, if one $x_i$ is stable, so is $x_1$. Likewise for the $y_j$. But $\beta^{x_1}$ and $\beta^{y_1}$ are non comparable by assumption, so by Lemma~\ref{lem:gen-Young} there are no stable state.

Let us now prove the claim in the general case, which holds if $n = m = 1$ for the same reason as in Figure~\ref{fig:no-stable1}, so let us assume that $n +m > 2$. Wlog let us also assume that any two products of $n'$ and $m'$ functions from $F$ with $n' + m' < n + m$ are comparable. So the functions in $F$ are pairwise comparable; (up to intersection of $I$ with a neighborhood of $0$) their supports constitute a linearly ordered set for the inclusion; and $\prod_i f_i \precsim \prod_{1 < j} g_j$ since $g_1 \leq 1$ and $\neg( \prod_j g_j \precsim \prod_{i} f_i)$. Up to renaming let us assume that $g_1$ has the smallest support $J$. Up to restriction of $I$ to the support of $\prod_i f_i$, which does not change stability nor non-comparability of $\prod_i f_i$ and $\prod_j g_j$ by Lemmas~\ref{lem:cong}.\ref{lem:cong8} and \ref{lem:cong}.\ref{lem:cong7}, let us assume that the $f_i$ are positive, and so is $g_j$ for $j > 1$ since $\prod_i f_i \precsim \prod_{1 < j} g_j$. If $g_1$ is also positive, we are back to the special case above, so let us assume that it is not. On the one hand, restricting $p$ to $I \setminus J$ shows that only $y_1$ might be stable, by Lemma~\ref{lem:cong}.\ref{lem:cong8}; on the other hand $\neg( \prod_j g_j\mid_J \precsim \prod_{i} f_i \mid_J)$, so let us make a case disjunction: if $\prod_i f_i\mid_J \precsim \prod_{j} g_j \mid_J$ then $\beta^{y_1} \in o(\beta^{x_1})$, so $y_1$ cannot be stable by Lemma~\ref{lem:gen-Young}; if they are not comparable, the special case above says that the $p\mid_J$ has no stable state, so neither has $p$ by Lemma~\ref{lem:cong}.\ref{lem:cong8}.
\end{proof}

\begin{proof}[Lemma~\ref{lem:hsr}]
Let us assume that $\mu$ is stationary, so it is well-known that its support involves only essential states. To prove that the equation holds let us make a case disjunction: if $x$ and $y$ are transient states, $\mu(x) = \mu(y) = 0$; if $x$ is transient and $y$ is essential, $\mu(x) = 0$ and $\p^y(\tau^+_x < \tau^+_y) = 0$; if $x$ and $y$ belong to distinct essential classes,  $\p^x(\tau^+_y < \tau^+_x) = \p^y(\tau^+_x < \tau^+_y) = 0$; if $x$ and $y$ belong to the same essential class $E$, let $E_1,\dots,E_k$ be the essential classes. Then for all $i \in\{1,\dots,k\}$ let  $\mu_{E_i}$ be the extension to $S$ (by the zero-function outside of $E_i$) of the unique stationary distribution of $p\mid_{E_i \times E_i}$.  So $\frac{\mu(x)}{\mu(y)} = \frac{\mu_E(x)}{\mu_E(y)}$ by  \cite[Proposition 2.1]{BL14} and since $\mu$ is a convex combination of the $\mu_{E_1},\dots,\mu_{E_k}$.

Conversely, let us assume that the support of $\mu$ involves only essential states and that the equation holds. Let $E'_1,\dots,E'_{k'}$ be the essential classes with positive $\mu$-measure. Let $i < k'$, and for all $x\in E'_i$ let $\mu_{i}(x) := \frac{\mu\mid_{E'_i}(x)}{\sum_{y\in E'_i}\mu(y)}$ define a distribution for $p\mid_{E'_i\times E'_i}$. Since $\mu_i$ also satisfies the equation and that $p\mid_{E'_i\times E'_i}$ is irreducible, $\mu_i$ is the unique stationary distribution for $p\mid_{E'_i\times E'_i}$. Since $\mu$ is a convex combination of the $\mu_{1},\dots,\mu_{k'}$, it is stationary for $p$.
\end{proof}

\begin{proof}[Lemma~\ref{lem:state-del}]
\begin{enumerate}
\item Let $\sigma_N := \sup\{n\in\N : |\{ k \leq n: X_k\in \tilde{S}\}| = N\}$ be the supremum of the first time that $(X_n)_{n\in\N}$ has visited $N$ states in $\tilde{S}$. Clearly $\sigma_N\stackrel{N\to\infty}{\longrightarrow} \infty$, so $\p^x(\tau^+_y < \tau^+_x) = \lim_{N\to\infty}\p^x(\tau^+_y < \min(\tau^+_x,\sigma_N))$. On the other hand, $\p^x(\tau^+_y < \min(\tau^+_x,\sigma_N))$

\begin{align*}
 & = \p^x(X_{\tau^+_{S\setminus T}} = y)\\
	& + \sum_{z\in S\setminus (T\cup\{x,y\})} \p^x(X_{\tau^+_{S\setminus T}} = z)\p^x(\tau^+_y < \min(\tau^+_x,\sigma_{N-1}))\\
	& = \tilde{p}(x,y) + \sum_{z\in S\setminus (T\cup\{x,y\})} \tilde{p}(x,z)\p^x(\tau^+_y < \min(\tau^+_x,\sigma_{N-1}))\\
	&\mbox{thus by iteration we obtain}\\
& = \sum_{k=1}^{N-1}\sum_{z_1.\dots,z_k\in S\setminus(T\cup\{x,y\})} \tilde{p}(x,z_1)\tilde{p}(z_1\dots z_k)\tilde{p}(z_k,y)\\
	& =\tilde{\p}^x(\tau^+_y < \min(\tau^+_x,\sigma_N)) \stackrel{N\to\infty}{\longrightarrow}\tilde{\p}^x(\tau^+_y < \tau^+_x)
\end{align*}

Thus $\p^x(\tau_y < \tau^+_x) = \tilde{\p}^x(\tau_y < \tau^+_x)$.

\item Let us first assume that $p$ is irreducible, and so is $\tilde{p}$. Let $\mu$ and $\tilde{\mu}$ be their respective unique, positive stationary distributions. By Lemmas~\ref{lem:hsr} and Lemma~\ref{lem:state-del}.\ref{lem:state-del1} we find 
\begin{equation*}\label{eq:ratio-mu-p-2}
\frac{\mu(y)}{\mu(x)} = \frac{\p^x(\tau^+_y < \tau^+_x)}{\p^y(\tau^+_x < \tau^+_y)} = \frac{\tilde{\p}^{x}(\tau^+_y < \tau^+_{x})}{\tilde{\p}^y(\tau^+_{x} < \tau^+_y)} = \frac{\tilde{\mu}(y)}{\tilde{\mu}(x)}
\end{equation*}
Summing this equation over $y\in \tilde{S}$ proves the irreducible case.

To prove the general claim, let $E_1\dots,E_{k}$ be the essential classes of $p$, so the essential classes of $\tilde{p}$ are the non empty sets among $E_1 \cap\tilde{S},\dots, E_k \cap\tilde{S}$. For $i \leq k$ let $\mu_i$ ($\tilde{\mu}_i$) be the extension to $S$ ($\tilde{S}$) of the unique stationary distribution of the irreducible $p\mid_{E_i\times E_i}$ ($\tilde{p}\mid_{E_i\cap \tilde{S}\times E_i\cap\tilde{S}}$). Let $\mu$ be a stationary distribution for $p$, it is well-known that $\mu$ is then a convex combination $\sum_{1\leq i \leq k}\alpha_i\mu_i$, and it is straightforward to check that the convex combination $\tilde{\mu} := \sum_{1\leq i \leq k}\frac{\sum_{y\in \tilde{S} \cap E_i}\mu(y)}{\sum_{y\in \tilde{S}}\mu(y)}\cdot \tilde{\mu}_i$ witnesses the claim. Conversely, let $\tilde{\mu}$ be a stationary distribution of $\tilde{p}$, so it is a convex combination $\sum_{1\leq i \leq k}\beta_i\tilde{\mu}_i$, and it is straightforward to check that the convex combination $\mu := \sum_{1 \leq i\leq k}\frac{L_i}{\sum_{j}L_j}\mu_i$ witnesses the claim, where $L_i := \beta_i \cdot \prod_{j\neq i} \frac{\mu_j(x_j)}{\tilde{\mu}_j(x_j)}$ for any $x_j\in \tilde{S}\cap E_j$.
\end{enumerate}
\end{proof}

\begin{proof}[Lemma~\ref{lem:congp-congmu}]
Let us first prove the claim for irreducible perturbations, and even in the following simpler case: let $x\in S$ be such that for all $y$ and all $z\neq x$ we have $p(x,z) \cong \tilde{p}(x,z)$ and $p(z,y) = \tilde{p}(z,y)$; so $\p^{z}(\tau_y < \tau^{+}_x) = \tilde{\p}^{z}(\tau_y < \tau^{+}_x)$ for all $z \neq x$ since the paths leading from $z$ to $y$ without hitting $x$ do not involve any step from $x$ to another state. So
\begin{align*}
\p^{x}(\tau^{+}_y < \tau^{+}_x)  &= \sum_{z\in S\setminus\{x\}}p(x,z)\p^{z}(\tau_y < \tau^{+}_x) \\
 & \cong \sum_{z\in S\setminus\{x\}}\tilde{p}(x,z)\tilde{\p}^{z}(\tau_y < \tau^{+}_x) = \tilde{\p}^{x}(\tau^{+}_y < \tau^{+}_x)
 \end{align*}
So by Lemmas~\ref{lem:hsr}, \ref{lem:cong}.\ref{lem:cong2}, and \ref{lem:cong}.\ref{lem:cong3}, and since the unique $\mu$ and $\tilde{\mu}$ are positive,
\[\frac{\mu(x)}{\mu(y)} = \frac{\p^{y}(\tau^{+}_x < \tau^{+}_y)}{\p^{x}(\tau^{+}_y < \tau^{+}_x)}  \cong \frac{\tilde{\p}^{y}(\tau^{+}_x < \tau^{+}_y)}{\tilde{\p}^{x}(\tau^{+}_y < \tau^{+}_x)} = \frac{\tilde{\mu}(x)}{\tilde{\mu}(y)}\mbox{\quad for all }y\in S;\]
invoking this equation above twice shows that
\[\frac{\mu(z)}{\mu(y)} = \frac{\mu(z)}{\mu(x)} \cdot \frac{\mu(x)}{\mu(y)}  \cong \frac{\tilde{\mu}(z)}{\tilde{\mu}(x)} \cdot \frac{\tilde{\mu}(x)}{\tilde{\mu}(y)} = \frac{\tilde{\mu}(z)}{\tilde{\mu}(y)}\mbox{\quad for all }z,y\in S;\]
and summing this second equation over $z\in S$ yields $\frac{1}{\mu(y)} \cong \frac{1}{\tilde{\mu}(y)}$, \textit{i.e.}, $\mu(y) \cong \tilde{\mu}(y)$ for all $y\in S$. The irreducible case is then proved by induction on $n := |\{x\in S : \exists y\in S,x\neq y \wedge p(x,y)\neq \tilde{p}(x,y)\}|$, which trivially holds for $n = 0$. For $0 < n$ let distinct $x,y\in S$ be such that $p(x,y)\neq \tilde{p}(x,y)$, and for all $y,z\in S \times S\setminus\{x\}$ let $\hat{p}(z,y) := p(z,y)$, and $\hat{p}(x,y) := \tilde{p}(x,y)$. By the simpler case $\mu \cong \hat{\mu}$; by induction hypothesis $\hat{\mu} \cong \tilde{\mu}$; so $\mu \cong \tilde{\mu}$ by transitivity of $\cong$. 

Let us now prove the general claim by induction on the number of the non-zero transition maps of $p$ that have zeros. Base case, all the non-zero maps are positive. Let $E'_1\dots,E'_{k'}$ be the sink SCCs of the graph on $S$ with arc $(x,y)$ if $p(x,y)$ is non-zero. For $i \leq k'$ let $\mu_i$ ($\tilde{\mu}_i$) be the unique stationary distribution map of the irreducible $p\mid_{E'_i\times E'_i}$ ($\tilde{p}\mid_{E'_i\times E'_i}$). Since $p(x,y)\mid_{E'_i\times E'_i} \cong \tilde{p}(x,y)\mid_{E'_i\times E'_i}$ for all $x \neq y$, the irreducible case implies $\mu_i \cong \tilde{\mu}_i$. Clearly $\mu$ is a convex combination of the $\mu_i$,and  the convex combination $\tilde{\mu}$ of the $\tilde{\mu}_i$ with the same coefficients is a stationary distribution map for $\tilde{p}$, and $\mu \cong \tilde{\mu}$ by Lemma~\ref{lem:cong}.\ref{lem:cong3}.

Inductive case. Let $p(z,t)$ be a non-zero function with support $J \subsetneq I$. Up to focusing we may assume that $0$ is a limit point of both $J$ and $I\setminus J$. By induction hypothesis on $p\mid_{J} \cong \tilde{p}\mid_{J}$ and $p\mid_{I\setminus J} \cong \tilde{p}\mid_{I\setminus J}$ we obtain two distribution maps $\tilde{\mu}_I$ and $\tilde{\mu}_{I\setminus J}$ that can be combined to witness the claim.
\end{proof}

\begin{proof}[Observation~\ref{obs:inf-bound}]
Let $p(x,y)$ be in the essential graph. By the definition of $\precsim$ and finiteness of the state space, let positive $b_{xy}$ and $\e_{xy}$ be such that $p(x,z) \leq b_{xy} \cdot p(x,y)$ for all $\e < \e_{xy}$ and $z\in S$. So for all $\e < \e_{xy}$ we have  $1 = \sum_{z\in S}p_\e(x,z) \leq |S| \cdot b_{xy} \cdot p_\e(x,y)$. Now let $b$ ($\e_0$) be the maximum (minimum) of the $b_{xy}$ ($e_{xy}$) for $(x,y)$ in the essential graph. Thus $0 < (b \cdot |S|)^{-|S|} \leq p_\e(\gamma)$ for all $\e < \e_0$.
\end{proof}

\begin{proof}[Lemma~\ref{lem:ess-weight}]
For all $y\in T$ let $x_y$ be an essential state reachable from $y$ in the essential graph. So $c\cdot \mu(y) \leq \mu(x_y)$ for all $y\in T$ by Lemma~\ref{lem:hsr}. Therefore $1 \leq c^{-1} \sum_{y\in T}\mu(x_y) + \sum_{x\in S\setminus T}\mu(x)$, and the claim is proved by further approximation.
\end{proof}

\noindent Lemma~\ref{lem:et} below is a technical tool proved by a standard argument in Markov chain theory. 

\begin{lemma}\label{lem:et}
Let a perturbation with state space $S$ satisfy Assumption~\ref {Assum2}, and let $x$ be in some essential class $E$. Then for all $n\in\mathbb{N}$
\[\p^x(\tau^+_{(S\setminus E)\cup\{x\}} > n) \leq  (1-c)^{\lfloor\frac{n}{|S|}\rfloor}\]

\begin{proof}
Let $\tau^* := \tau^+_{(S\setminus E)\cup\{x\}}$. For every $y\in E$ let $\gamma_y\in\Gamma_E(y,x)$ be in the essential graph. $c < p(\gamma_y)$ by Assumption~\ref{Assum2}, so $\max_{y\in E} \p_{\e}^y (\tau^* > |S|) \leq 1 - c < 1$, so for all $k\in\mathbb{N}$
\begin{align*}
\p^x (\tau^* > k|S|) & \leq (\max_{y\in E} \p^y (\tau^* > |S|)^k \\
 & \leq (1-c)^k \textrm{ by the strong Markov property, so}\\
\p^x (\tau^* > n)  &\leq \p^x (\tau^* > |S|\cdot\lfloor\frac{n}{|S|}\rfloor) \leq (1-c)^{\lfloor\frac{n}{|S|}\rfloor} \textrm{ for all }n\in\mathbb{N}.
\end{align*}
\end{proof}
\end{lemma}

\begin{proof}[Lemma~\ref{lem:p-cong-max}]
Up to focusing let $p$ satisfy Assumption~\ref{Assum2}. The second statement boils down to Lemma~\ref{lem:cong}.\ref{lem:cong3}. For the first statement, for every $z\in E$ let $\gamma_z\in \Gamma_E(x,z)$ be a path in the essential graph, so $c < p(\gamma_z)$ by Assumption~\ref{Assum2}. Let $\tau^* := \tau^+_{(S\setminus E)\cup\{x\}}$, so
\begin{align*}
c \cdot p(z,y) & \leq  c \cdot \p^z(X_{\tau^*} = y) \leq p(\gamma_z) \cdot \p^z(X_{\tau^*} = y) \\
 & \leq \p^x(X_{\tau^*} = y) =  \tilde{p}(\cup E,y)
 \end{align*}
which implies $c \cdot \max_{z\in E}p(z,y) \leq \tilde{p}(\cup E,y)$ \textit{i.e.} half of the statement $\tilde{p}(\cup E,y) \cong \max_{z\in E} p(z,y)$. Now, for every positive natural $N$ let ${\bf A}_N := \{x\} \times (E\setminus\{x\})^{N-1} \times \{y\}$.
\begin{align*}
\p^x(X_{\tau^*} = y,\tau^* = N) & = \sum_{\gamma\in {\bf A}_N} p(\gamma)\\
	& \leq \max_{z\in E}p(z,y)\sum_{\gamma t\in {\bf A}_N} p(\gamma)\\
	& = \max_{z\in E}p(z,y) \p^x(\tau^* \geq N-1)
\end{align*}
\noindent Let $q := (1-c)^{|S|} < 1$, so $\p^x(\tau^* \geq N) \leq (1-c)^{-1}\cdot q^N$ by Lemma~\ref{lem:et}, and
\begin{align*}
\p^x(X_{\tau^*} = y) & = \sum_{N=1}^{\infty}\p^x(X_{\tau^*} = y,\tau^* = N)\\
	& \leq \max_{z\in E}p(z,y)\sum_{N=0}^{\infty}(1-c^{-1})\cdot q^N \\
	& = (1-c)^{-1}(1-q)^{-1}\cdot \max_{z\in E}p(z,y)
\end{align*}
\end{proof}

\begin{proof}[Proposition~\ref{prop:essential-graph}]
Up to focusing let $p$ satisfy Assumption~\ref{Assum2}.
\begin{enumerate}
\item Let $y$ be in the set of the transient states $T$, so there exists $x\notin T$ and $\gamma\in \Gamma_T(y,x)$ such that $\gamma$ is also in the essential graph. By Assumption~\ref{Assum2} this implies
\[0 < \liminf_{\e\to 0}p_\e(\gamma) \leq \liminf_{\e\to 0}\p_{\e}^y(\tau^+_x < \tau^+_y).\]
On the other hand, let $E$ be the essential class of $x$. Lemma~\ref{lem:p-cong-max} implies that
\begin{align*}
\p^x(\tau^+_y < \tau^+_x) & \leq 1- \p^x(X_{\tau^+_{S\setminus E\cup\{x\}}} = x) \\
&= \sum_{z\notin E}\p^x(X_{\tau^+_{S\setminus E\cup\{x\}}} = z) \cong \sum_{z\notin E}\max_{t\in E}p(t,z)
\end{align*}
Since $E$ is an essential class, for all $t\in E$ and $z\notin E$ the function $p(t,z)$ is not $\precsim$-maximal. So $\liminf_{\e\to 0} p_\e(t,z) = 0$ by definition of the essential graph, and $\liminf_{\e\to 0} \sum_{z\notin E}\max_{t\in E}p_\e(t,z) = 0$ by Assumption~\ref{Assum1}, Lemma~\ref{lem:cong}.\ref{lem:cong6}, and finiteness. By combining this with the two inequalities above one obtains
\[\liminf_{\e\to 0} \frac{\p_{\e}^x(\tau^+_y < \tau^+_x)}{\p_{\e}^y(\tau^+_x < \tau^+_y)} = 0\]
and noting the following by Lemma~\ref{lem:hsr} allows us to conclude.
\[\tilde{\mu}_{\e}(y) \leq \frac{\tilde{\mu}_{\e}(y)}{\mu_{\e}(x)} = \frac{\p_{\e}^x(\tau^+_y < \tau^+_x)}{\p_{\e}^y(\tau^+_x < \tau^+_y)}\]

\item By Assumption~\ref{Assum2} (twice) and Lemma~\ref{lem:hsr}.
\end{enumerate}
\end{proof}

\begin{proof}[Lemma~\ref{lem:preserve-stable}]
Let us first prove the first conjunct. On the one hand $\p^y(\tau_x < \tau_y) \leq \p^y(\tau_E < \tau_y) = \tilde{\p}^y(\tau_{\tilde{x}} < \tau_y)$, and on the other hand
$\p^y(\tau_x < \tau_y) \geq  \p^y(\tau_E < \tau_y)\cdot \min_{z\in E}\p^z(\tau_x < \tau_y)$ by the strong Markov property. Since $\lim_{\e\to 0}\p^z(\tau_x < \tau_y) = 1$ for all $z$ in the essential class $E$, $\tilde{\p}^y(\tau_{\tilde{x}} < \tau_y)$ and $\p^y(\tau_x < \tau_y)$ are even asymptotically equivalent.

Let us now prove the second conjunct.  Let $\tau^*_{E,0} := 0$ and $\tau^*_{E,n+1} := \inf \{t > \tau^*_{E,n}\, \mid\, X_t\in E\,\wedge\,\exists s\in ]\tau^*_{E,n},t[,\,X_s\notin E\}$ for all $n\in \mathbb{N}$. Informally, when starting in $E$, $\tau^*_{E,n}$ is the first time that the chain is in $E$ after $n$ stays outside that are separated by stays inside.
\begin{align*}
\tilde{\p}^{\cup E}(\tau_y < \tau_{\cup E}) & = \p^x(\tau_y < \tau_x, \tau_y <\tau^*_{E,1}) \\
& \textrm{ by definition of the essential collapse, and } \tau^*_{E,1}\\
	& \leq \p^x(\tau_y < \tau_x) \textrm{, showing half of the first conjunct.} 
\end{align*}
For the other half, let $X$ satisfy Assumption~\ref {Assum2} up to focusing. For all $z\in E \setminus \{x\}$ there is a simple path in the essential graph from $x$ to $z$, so $c < \p^x(\tau_z < \tau_{S\setminus E \cup \{x\}})$, and by the strong Markov property
\begin{align}
\p^x(\tau_y < \tau_x,\tau_y < \tau^*_{E,1}) & \geq \p^x(\tau_z < \tau_{S\setminus E \cup \{x\}})\cdot \p^z(\tau_y < \tau_x,\tau_y < \tau^*_{E,1}) \nonumber\\
	& \geq c \cdot \p^z(\tau_y < \tau_x,\tau_y < \tau^*_{E,1}) \label{ineq1}
\end{align}
For all $z\in E\setminus \{x\}$ there is also a simple path in the essential graph from $z$ to $x$, so $\p^z(\tau_{S\setminus E} < \tau_x) \leq 1 - c$. Also note that $\p^z(\tau^*_{E,1} < \tau_x) \leq \p^z(\tau_{S\setminus E} < \tau_x)$, and let us show by induction on $n$ that $\p^x(\tau^*_{E,n} < \tau_x) \leq (1-c)^n$, which holds for $n =0$.
\begin{align}
\p^x(\tau^*_{E,n+1} < \tau_x) &\leq \p^x(\tau^*_{E,n} < \tau_x) \cdot \max_{z\in E} \p^z(\tau^*_{E,1} < \tau_x) \\
 & \textrm{ by the strong Markov property.} \nonumber\\
	& \leq \p^x(\tau^*_{E,n} < \tau_x) \cdot (1-c) \textrm{ by the remark above.}\nonumber\\
	& \leq (1-c)^{n+1} \textrm{ by induction hypothesis.} \label{ineq2}
\end{align}
Let us now conclude about the second half of the first conjunct.
\begin{align*}
\p^x(\tau_y < \tau_x) & = \sum_{n=0}^\infty \p^x(\tau_y < \tau_x, \tau^*_{E,n} < \tau_y < \tau^*_{E,n+1})\\
& \textrm{ by a case disjunction.}\\
	& \leq \sum_{n=0}^\infty \p^x(\tau_y < \tau_x, \tau^*_{E,n} < \tau_x, \tau_y < \tau^*_{E,n+1})\\
	& \textrm{ since the new conditions are weaker.}\\
	& \leq \sum_{n=0}^\infty \p^x(\tau^*_{E,n} < \tau_x) \max_{z\in E}\p^z(\tau_y < \tau_x, \tau_y < \tau^*_{E,1})\\
	& \textrm{ by the strong Markov property.}\\
	& \leq c^{-1}\p^x(\tau_y < \tau_x, \tau_y < \tau^*_{E,1}) \cdot \sum_{n=0}^\infty (1-c)^n\\
	& \textrm{ by inequalities~\ref{ineq1} and \ref{ineq2}.}\\
	& \leq c^{-2} \tilde{\p}^{\cup E}(\tau_y < \tau_{\cup E})
\end{align*}
\end{proof}

\begin{proof}[Proposition~\ref{prop:essential-trans}]
Up to focusing let $p$ satisfy Assumption~\ref{Assum2}.
\begin{enumerate}
\item Let us first assume that $p$ is irreducible, and so is $\tilde{p}$ by Observation~\ref{obs:ess-coll}. So both $p$ and $\tilde{p}$ have unique, positive stationary distribution maps. Let us prove that $\tilde{\mu}(\tilde{x}) \cong \mu(x)$ where $\tilde{x} := \cup E$ and $\tilde{\mu}(y) \cong \mu(y)$ for all $y\in S\setminus E$. For $y\in S\setminus E$, Lemma~\ref{lem:hsr} and Lemma~\ref{lem:preserve-stable} imply the following. 
\begin{equation}\label{eq:ratio-mu-p-2}
\frac{\mu(y)}{\mu(x)} = \frac{\p^x(\tau^+_y < \tau^+_x)}{\p^y(\tau^+_x < \tau^+_y)} \cong \frac{\tilde{\p}^{\tilde{x}}(\tau^+_y < \tau^+_{\tilde{x}})}{\tilde{\p}^y(\tau^+_{\tilde{x}} < \tau^+_y)} = \frac{\tilde{\mu}(y)}{\tilde{\mu}(\tilde{x})}
\end{equation}
Summing the above equation over $y\in S\setminus E$ and adding $\frac{\mu(x)}{\mu(x)} = \frac{\tilde{\mu}(\tilde{x})}{\tilde{\mu}(\tilde{x})}$ yields $(1-\bar{\mu})\mu(x)^{-1} \cong \tilde{\mu}(x)^{-1}$, where $\bar{\mu} := \sum_{z\in E\setminus\{x\}}\mu(z)$. So by Definition~\ref{def:cong}, let $a$ and $b$ be positive real numbers such that $a\cdot \tilde{\mu}(\tilde{x})^{-1} \leq (1-\bar{\mu})\mu(x)^{-1} \leq b\cdot \tilde{\mu}(\tilde{x})^{-1}$ on a neighborhood of $0$, which yields $a\cdot \mu(x) \leq \tilde{\mu}(\tilde{x}) \leq b\cdot \mu(x) + \bar{\mu}$. Since $\bar{\mu} \leq b'\cdot \mu(x)$ for some $b'$ by Proposition~\ref{prop:essential-graph}.\ref{prop:essential-graph2}, $\tilde{\mu}(\tilde{x}) \cong \mu(x)$. Now by Lemmas~\ref{lem:cong}.\ref{lem:cong2} and \ref{lem:cong}.\ref{lem:cong3}, let us rewrite $\mu(x)$ with $\tilde{\mu}(\tilde{x})$ in Equation~\ref{eq:ratio-mu-p-2}, which shows the claim for irreducible perturbations.

Let us now prove the general claim by induction on the number of the non-zero transition maps of $p$ that have zeros. Base case, all the non-zero maps are positive. Let $E'_1\dots,E'_{k'}$ be the sink SCCs of the graph on $S$ with arc $(x,y)$ if $p(x,y)$ is non-zero. The essential graph is included in this digraph, and the essential class $E$ is included in $E'_j$ for some $j\in\{1,\dots,k'\}$. Let $\tilde{E}'_j := \{E\} \cup E'_j\setminus E$ and for all $i \neq j$ let $\tilde{E}'_i := E'_i$. For all $i\in\{1,\dots,k'\}$ let $\mu_{E'_i}$ ($\tilde{\mu}_{\tilde{E}'_i}$) be the extension to $S$ ($\{E\}\cup S\setminus E$), by the zero-function outside of $E'_i$ ($\tilde{E}'_i$), of the unique stationary distribution of $p\mid_{E'_i \times E'_i}$ ($\tilde{p}\mid_{\tilde{E}'_i \times \tilde{E}'_i}$), and by the irreducible case above let $\tilde{\mu}_{\tilde{E}'_j}$ ($\mu_{E'_i}$) be the corresponding unique distribution of $\tilde{p}\mid_{\{E\}\cup E'_j\setminus E}$ after ($p\mid_{E'_j\times E'_j}$ before) the essential collapse. Since $\mu$ ($\tilde{\mu}$) is a convex combination $\sum_{1\leq i \leq k'}\alpha_i\mu_{E'_i}$ ($\sum_{1\leq i \leq k'}\alpha_i\tilde{\mu}_{\tilde{E}'_i}$), it is easy to check that $\tilde{\mu} := \sum_{1\leq i \leq k'}\alpha_i\tilde{\mu}_{\tilde{E}'_i}$ ($\mu := \sum_{1\leq i \leq k'}\alpha_i\mu_{E'_i}$) is a witness for the base case.

Inductive case. Let $p(z,t)$ be a non-zero function with support $J \subsetneq I$. Up to focusing we may assume that $0$ is a limit point of both $J$ and $I\setminus J$. By induction hypothesis on $p\mid_{J} \cong \tilde{p}\mid_{J}$ and $p\mid_{I\setminus J} \cong \tilde{p}\mid_{I\setminus J}$ we obtain two distribution maps $\tilde{\mu}_I$ and $\tilde{\mu}_{I\setminus J}$ ($\mu_I$ and $\mu_{I\setminus J}$) that can be combined to witness the claim.

\item By Propositions~\ref{prop:essential-graph}.\ref{prop:essential-graph1} and \ref{prop:essential-trans}.\ref{prop:essential-trans3} in both cases, and also by Proposition~\ref{prop:essential-graph}.\ref{prop:essential-graph2} if $y\in E$. 
\end{enumerate}
\end{proof}

\begin{proof}[Proposition~\ref{prop:trans-del}]
Up to focusing let $p$ satisfy Assumption~\ref{Assum2}. By induction on the number of the non-zero transition maps of $p$ that have zeros. Base case, all the non-zero maps are positive. Let $E'_1\dots,E'_{k'}$ be the sink SCCs of the graph on $S$ with arc $(x,y)$ if $p(x,y)$ is non-zero. Note that this digraph includes the essential graph, and that $\delta(p\mid_{E'_i \times E'_i}) = \delta(p)\mid_{E'_i\cap S\setminus T \times E'_i\cap S\setminus T}$. Moreover, by Lemmas~\ref{lem:state-del}.\ref{lem:state-del2} the stable states of each $p\mid_{E'_i \times E'_i}$ are also stable for $\delta(p\mid_{E'_i \times E'_i})$, and the converse holds by Lemmas~\ref{lem:state-del}.\ref{lem:state-del2} and \ref{lem:ess-weight}. Therefore a state is stable for $p$ iff it is stable for some $p\mid_{E'_i \times E'_i}$ iff it is stable for some $\delta(p)\mid_{E'_i\cap S\setminus T \times E'_i\cap S\setminus T}$ iff it is stable for $\delta(p)$.

Inductive case. Let $p(z,t)$ be a non-zero function with support $J \subsetneq I$. Up to focusing we may assume that $0$ is a limit point of both $J$ and $I\setminus J$. By induction hypothesis $p\mid_{J}$ and $\delta(p)\mid_{J}$ have the same stable states, and likewise for $p\mid_{I\setminus J}$ and $\delta(p)\mid_{I\setminus J}$, which shows the claim.
\end{proof}

Lemma~\ref{lem:escape-decomp} below is a generalization to the reducible case of Proposition 2.8 from~\cite{BL14}.

\begin{lemma}\label{lem:escape-decomp}
Let $A$ be a finite subset of the state space $S$ of a Markov chain. Then for all  $x\in A$ and $y\in S\setminus A$
\[\begin{array}{l}
\p^x(X_{\tau_{S\setminus A}} = y) =\\
\qquad \sum_{\gamma\in \Gamma_A(x,y), p(\gamma) > 0} \prod_{i=1}^{|\gamma| -1} \frac{p(\gamma_i,\gamma_{i+1})}{1-\p^{\gamma_i}(X_{\tau^+_{(S\setminus A)\cup \{\gamma_1,\dots,\gamma_i\}}} = \gamma_i)}
\end{array}\]

\begin{proof}
We proceed by induction on $|A|$. The claim trivially holds for $|A| = 0$; so now let $x \in A$. The strong Markov property gives
\begin{align*}
\p^x(X_{\tau_{S \setminus A}} = y)  &= \p^x(X_{\tau^+_{(S\setminus A)\cup\{x\}}} = y) \\
 & + 
\p^x(X_{\tau^+_{(S\setminus A)\cup\{x\}}} = x) \cdot \p^x(X_{\tau_{S\setminus A}} = y)
\end{align*}
If $p(\gamma) = 0$ for all $\gamma\in \Gamma_A(x,y)$, the claim boils down to $0 = 0$, so let us assume that there exists $\gamma\in\Gamma_A(x,y)$ with $p(\gamma) = 0$, so $\p^x(X_{\tau^+_{(S\setminus A)\cup\{x\}}} = x) < 1$, and the above equation may be rearranged into 
\[
\p^x(X_{\tau_{S\setminus A}} = y) = \frac{\p^x(X_{\tau^+_{(S\setminus A)\cup\{x\}}} = y)}{1-\p^x(X_{\tau^+_{(S\setminus A)\cup\{x\}}} = x)}
\]
where the numerator may be decomposed as\\$p(x,y) + \sum_{z\in A\setminus\{x\}}p(x,z)  \p^z(X_{\tau_{(S\setminus A)\cup\{x\}}} = y)$. By the induction hypothesis for the set $A \setminus \{x\}$ let us rewrite 
$\p^z(X_{\tau_{(S\setminus A) \cup\{x\}}} = y)$ for all $z \in A\setminus\{x\}$, and obtain the equation below that may be re-indexed to yield the claim.
\begin{align*}
 \p^x(X_{\tau_{S\setminus A}} = y)  &= \frac{p(x,y)}
{1-\p^x(X_{\tau^+_{(S\setminus A)\cup\{x\}}} = x)} + \sum_{z \in A \setminus \{ x \}, p(x,z) > 0}\\
	& \sum_{ \gamma\in \Gamma_{A\setminus\{x\}}(z,y), p(\gamma) > 0} 
	\frac{p(x,z)\cdot \Pi} {1-\p^x(X_{\tau^+_{(S\setminus A)\cup\{x\}}} = x)} \\
	& \textrm{where } \Pi := \prod_{i=1}^{|\gamma|-1} \frac{p(\gamma_i,\gamma_{i+1})}{1-\p^{\gamma_i}(X_{\tau^+_{(S\setminus A)\cup 
	\{x,\gamma_1,\dots,\gamma_i\}}} = \gamma_i)}
\end{align*}
\end{proof}
\end{lemma}

\begin{proof}[Lemma~\ref{lem:singleton-essential}]
Up to focusing let $p$ satisfy Assumption~\ref{Assum2}. Let $\tau^* := \tau^+_{S\setminus T}$, and consider
\[\p^x(X_{\tau^*} = y) = p(x,y) + \sum_{z\in T}p(x,z)\p^z(X_{\tau^*} = y).\]
For all $z\in T$ there are $z'\in S\setminus T$ and $\gamma_z\in\Gamma_T(z,z')$ in the essential graph, so  for all $K \subseteq S$ we have $\p^z(X_{\tau^+_{(S\setminus T)\cup K}} = z) \leq \p^z(X_{\tau^+_{(S\setminus T)\cup\{z\}}} = z) \leq 1 - p(\gamma) < 1-c$. So by Lemma~\ref{lem:escape-decomp}
\[\p^z(X_{\tau^*} = y) \leq \sum_{\gamma\in \Gamma_T(z,y)}p(\gamma)\cdot c^{-|\gamma|} \leq c^{-|T|}\sum_{\gamma\in \Gamma_T(z,y)}p(\gamma).\] Since $\p^z(X_{\tau^*} = y) \geq \sum_{\gamma\in \Gamma_T(z,y)}p(\gamma)$, thus $\p^z(X_{\tau^*} = y) \cong \sum_{\gamma\in \Gamma_T(z,y)}p(\gamma)$, an by Lemma~\ref{lem:cong}.\ref{lem:cong5} we can replace the sum with the maximum. 
\end{proof}

\begin{proof}[Observation~\ref{obs:prec-div}]
Let $\e\in I$. If $g(\e) \neq 0$ then $((f \div_n g)\cdot g) (\e) = \frac{f(\e)}{g(\e)}\cdot g(\e)= f(\e)$; if $g(\e) = 0$ then $f(\e) = 0 = (f \div_n g)(\e)\cdot g(\e)$.
\end{proof}

\begin{proof}[Proposition~\ref{prop:ogs}]
Let $m$ be as in Definition~\ref{defn:os} and let $J \subseteq I$ be its support.
\begin{enumerate}
\item 
$\sigma(p)(x,y)\mid_{I\setminus J} = |S|^{-1}$ for all $x,y\in S$, so $\sigma(p)_\e$ is stochastic for all $\e\in I\setminus J$. Let $\e\in J$. On the one hand $\sum_{y\in S}\sigma(p)_\e(x,y) = \frac{p_\e(x,x) + m - 1}{m} + \sum_{y\in S\setminus\{x\}} \frac{p_\e(x,y)}{m} = 1$, on the other hand $\sum_{y\in S\setminus\{x\}}\sigma(p)_\e(x,y) = \sum_{y\in S\setminus\{x\}} \frac{p_\e(x,y)}{|S|\cdot \max\{p_\e(t,z)\,\mid\,z,t\in S \wedge z \neq t\}}  \leq \frac{(|S|-1)}{|S|} < 1$, so $\sigma(p)_\e$ is stochastic. If $0$ is not a limit point of $J$, it is clear that $\sigma(p)(x,y) \cong 1$ for all $x,y\in S$, and $\sigma(p)$ satisfies Assumption~\ref{Assum1}. If $0$ is a limit point of $J$, then $p(x,y)\mid_J \precsim p(x',y')\mid_J$ implies $\sigma(p)(x,y)\mid_J \precsim \sigma(p)(x',y')\mid_J$ (since $m$ is non-zero on $J$), so $\sigma(p)$ satisfies Assumption~\ref{Assum1} by Lemma~\ref{lem:cong}.\ref{lem:cong7} (since $p\mid_J$ and then $\sigma(p)\mid_{J}$ do).

\item Let us first prove the claim if $J = I$. The equation below shows that $\mu \cdot p = \mu$ iff $\mu\cdot \sigma(p) = \mu$, for all distribution maps $\mu$ on $S = \{x_1,\dots,x_n\}$, so $p$ and $\sigma(p)$ have the same stable states. 
\begin{align*}
(\mu\cdot \sigma(p))_j & = \frac{p(x_j,x_j) + m - 1}{m}\cdot \mu_j + \sum_{i\neq j} \frac{p(x_i,x_j) \mu_i}{m} \\
&= \frac{(m-1)\mu_j + \sum_{i}p(x_i,x_j) \mu_i}{m}
\end{align*}
Let us now prove the claim if $J \subsetneq I$. The case where $0$ is not a limit point of $I\setminus J$ amounts, up to focusing, to the case $J =I$, so let us assume that $0$ is a limit point of $I\setminus J$. For all $\e\in I\setminus J$ we have $p_\e(x,y) = 0$ for all states $x \neq y$, so all distributions are stationary for $p_\e$. Moreover, the uniform distribution is stationary for $\sigma(p)_\e$, so, all states are stable for $p\mid_{I\setminus J}$ and $\sigma(p)\mid_{I\setminus J}$. Therefore, if $0$ is not a limit point of $J$, we are done, and if it is, we are also done by Lemma~\ref{lem:cong}.\ref{lem:cong8}.

\item Let distinct $x,y\in S$ be such that $p(z,t) \precsim p(x,y)$ for all distinct $z,t\in S$. So $p(x,y) \cong \max\{p(z,t)\,\mid\, z,t\in S \wedge z \neq t\}$ by Lemma~\ref{lem:cong}.\ref{lem:cong6}, so $\sigma(p)(x,y) \cong p(x,y) \div_{|S|} |S| \cdot p(x,y) = \frac{1}{|S|} \cong 1$, so $(x,y)$ is in the essential graph of $\sigma(p)$.
\end{enumerate}
\end{proof}

\begin{lemma}\label{lem:3t}
Let a perturbation $p$ with state space $S$ satisfy Assumption~\ref{Assum1}, let $E_1,\dots, E_k$ be its essential classes, and for all $i$ let $x_i \in E_i$. The state $x$ is stable for $p$ iff $x$ belongs to some $E_i$ such that $\cup E_i$ is stable for $\delta\circ\kappa(\dots\kappa(\kappa(\sigma(p),x_1),x_2)\dots,x_k)$.
\end{lemma}

\begin{proof}[Theorem~\ref{thm:stable-states}]
By applying Lemma~\ref{lem:3t} recursively. If $p$ is the identity matrix then all states are stable. Otherwise the essential graph of $\sigma(p)$ is non-empty, so either one essential class is not a singleton, or one state is transient. If there is a non-singleton essential class, the corresponding essential collapse decreases the number of states; if one state is transient, the transient deletion decreases the number of states. Since these transformations do not increase the number of states, $\delta\circ\kappa(\dots \kappa(\kappa(\sigma(p),x_1),x_2)\dots,x_k)$ has fewer states than $p$, whence termination of the recursion on an identity perturbation whose non-empty state space corresponds to the stable states of $p$.
\end{proof}

\begin{proof}[Observation~\ref{obs:div-semiring}]
If $f \leq 0$ then $f = (f \div 0)\cdot 0 = 0$. Also, $(f \div 1) = (f \div 1) \cdot 1 = f$.
\end{proof}

\begin{proof}[Lemma~\ref{lem:odsm-f}]
\begin{enumerate}
\item By Lemma~\ref{lem:cong}.\ref{lem:cong7} with $J$ the support of $g$ and $I\setminus J$, then by Lemmas~\ref{lem:cong}.\ref{lem:cong2} and  \ref{lem:cong}.\ref{lem:cong3}.

\item By Observations~\ref{obs:ccc-loc}.\ref{obs:ccc-loc2} and~\ref{obs:prec-div}, and since $\eqclass{\e\mapsto 1}$ is the $\eqclass{\precsim}$-maximum of $\eqclass{G\cup\{\e\mapsto 0\}}$.
\end{enumerate}
\end{proof}

\begin{proof}[Lemma~\ref{lem:pop}]
\begin{enumerate}
\item Clear by comparing Definitions~\ref{defn:essential-class} and \ref{defn:ao}.\ref{defn:ao1}.


\item Let $\eqclass{p}(z,t) = \max\{\eqclass{p}(z',t') : (z',t')\in S\times S \wedge z'\neq t'\}$. If $x \neq y$ then
\begin{align*}
\eqclass{\sigma}(\eqclass{p})(x,y) & = \eqclass{p}(x,y) \eqclass{\div} \eqclass{p}(z,t) \textrm{ by definition of }\eqclass{\sigma},\\
	& = \eqclass{p(x,y)} \eqclass{\div} \eqclass{p(z,t)} \textrm{ by definition of }\eqclass{p},\\
	& = \eqclass{p(x,y) \div_{|S|} p(z,t)} \textrm{ by Lemma~\ref{lem:odsm-f}},\\
	& =  \eqclass{p(x,y) \div_{|S|} \max\{p(z,t)\,\mid\,z,t\in S\wedge z \neq t\}}\\
	&\textrm{ by Lemma~\ref{lem:odsm-f} and Lemma~\ref{lem:cong}.\ref{lem:cong6}},\\
	& = \eqclass{\sigma(p)}(x,y) \textrm{ by definition of }\sigma.
\end{align*}

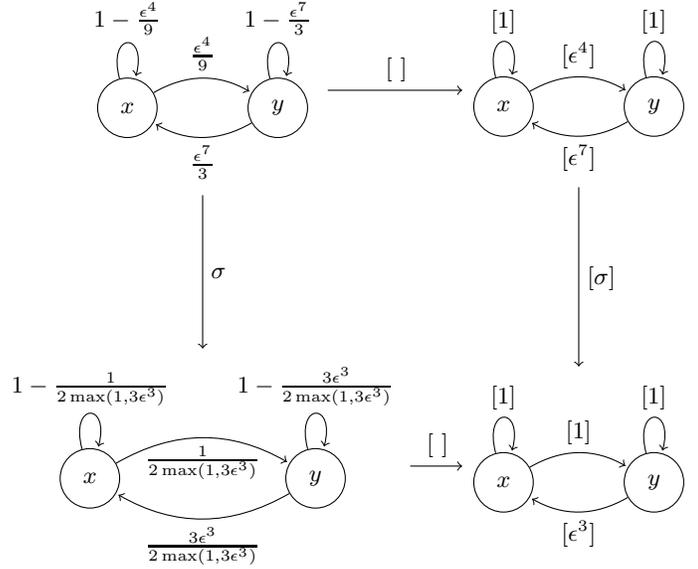
\begin{tikzpicture}[shorten >=1pt,node distance=5cm, auto]

\node (a) {
\begin{tikzpicture}[shorten >=1pt,node distance=2cm, auto]
  \node[state] (x) {$x$};
  \node[state] (y) [right of = x] {$y$};

\path[->] (x) edge [loop above] node {$1-\frac{\e^4}{9}$} ()
		edge [bend left] node {$\frac{\e^4}{9}$} (y)
		(y) edge [loop above] node {$1-\frac{\e^7}{3}$} ()
		edge [bend left] node {$\frac{\e^7}{3}$} (x);
 \end{tikzpicture}
};

\node (b) [right of = a]{
\begin{tikzpicture}[shorten >=1pt,node distance=2cm, auto]
  \node[state] (x) {$x$};
  \node[state] (y) [right of = x] {$y$};

\path[->] (x) edge [loop above] node {$\eqclass{1}$} ()
		edge [bend left] node {$\eqclass{\e^4}$} (y)
		(y) edge [loop above] node {$\eqclass{1}$} ()
		edge [bend left] node {$\eqclass{\e^7}$} (x);
 \end{tikzpicture}};
 
\node (c) [below of = b] {
\begin{tikzpicture}[shorten >=1pt,node distance=2cm, auto]
  \node[state] (x) {$x$};
  \node[state] (y) [right of = x] {$y$};

\path[->] (x) edge [loop above] node {$\eqclass{1}$} ()
		edge [bend left] node {$\eqclass{1}$} (y)
		(y) edge [loop above] node {$\eqclass{1}$} ()
		edge [bend left] node {$\eqclass{\e^3}$} (x);
 \end{tikzpicture}};
 
\node (d) [left of = c]{
\begin{tikzpicture}[shorten >=1pt,node distance=3cm, auto]
  \node[state] (x) {$x$};
  \node[state] (y) [right of = x] {$y$};

\path[->] (x) edge [loop above] node {$1-\frac{1}{2\max(1,3\e^3)}$} ()
		edge [bend left] node  [below]{$\frac{1}{2\max(1,3\e^3)}$} (y)
		(y) edge [loop above] node {$1-\frac{3\e^3}{2\max(1,3\e^3)}$} ()
		edge [bend left] node{$\frac{3\e^3}{2\max(1,3\e^3)}$} (x);
 \end{tikzpicture}};

\draw[->] (a) to node {$\eqclass{\,\,}$} (b);
\draw[->] (b) to node {$\eqclass{\sigma}$} (c);
\draw[->] (a) to node {$\sigma$} (d);
\draw[->] (d) to node {$\eqclass{\,\,}$} (c);
 \end{tikzpicture}

\item The essential classes of $p$ are singletons since $\delta(p)$ is well-defined. Let $\{x\}$ and $\{y\}$ be distinct essential classes of $p$, and of $\eqclass{p}$ by Lemma~\ref{lem:pop}.\ref{lem:pop0}. Let $M := \max_{\eqclass{\precsim}}\{\eqclass{p}(\gamma) : \gamma\in \Gamma_T(x,y)\}$, so $M = \eqclass{\max\{p(\gamma) : \gamma\in \Gamma_T(x,y)\}}$ by Lemma~\ref{lem:cong}.\ref{lem:cong6}. Note that $p(x,x) \cong 1$ since $\{x\}$ is an essential class of $p$, so $\sum_{z\in S\setminus T}\max\{p(\gamma) : \gamma\in \Gamma_T(x,z)\} \cong 1$ too. So 

\begin{align*}
&\eqclass{\chi}(\eqclass{p})(\{x\},\{y\}) = M \textrm{ by definition of }\eqclass{\chi},\\
&	 =  M \eqclass{\div} \eqclass{1} \textrm{ by Observation~\ref{obs:div-semiring}},\\
&	 = M \eqclass{\div} \eqclass{\sum_{z\in S\setminus T}\max\{p(\gamma) : \gamma\in \Gamma_T(x,z)\}}\textrm{ by a remark above},\\
&	 =  \eqclass{\max\{p(\gamma) : \gamma\in \Gamma_T(x,y)\}} \eqclass{\div} \\
&\eqclass{\sum_{z\in S\setminus T}\max\{p(\gamma) : \gamma\in \Gamma_T(x,z)\}}\textrm{ by a remark above},\\
&	 = \eqclass{\max\{p(\gamma) : \gamma\in \Gamma_T(x,y)\} \div_{|S|} \\
&\sum_{z\in S\setminus T}\max\{p(\gamma) : \gamma\in \Gamma_T(x,z)\}} \textrm{ by Lemma~\ref{lem:odsm-f}}, \\
&	 = \eqclass{\delta(p)}(x,y) \textrm{ by definition of }\delta.
\end{align*}

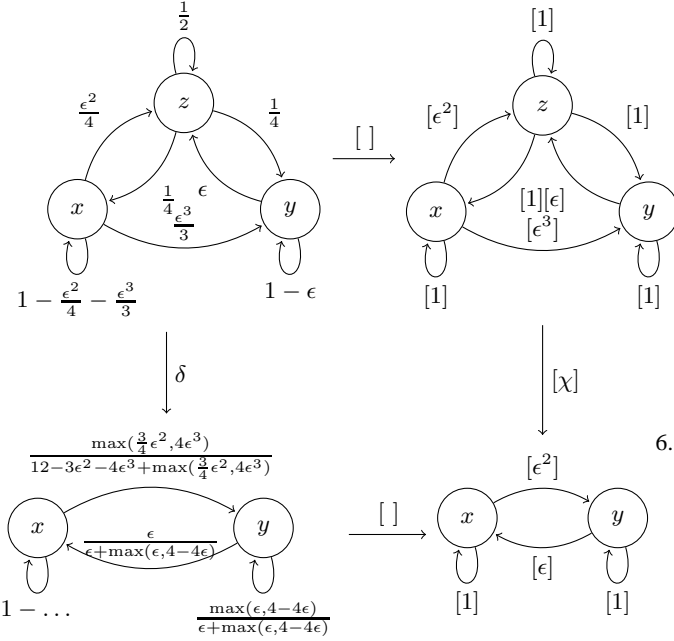
\begin{tikzpicture}[shorten >=1pt,node distance=5cm, auto]

\node (a) {
\begin{tikzpicture}[shorten >=1pt,node distance=2cm, auto]
  \node[state] (x) {$x$};
  \node[state] (z) [above right of = x] {$z$};
  \node[state] (y) [below right of = z] {$y$};

\path[->] (x) edge [loop below] node {$1-\frac{\e^2}{4}-\frac{\e^3}{3}$} ()
		edge [bend left] node {$\frac{\e^2}{4}$} (z)
		edge [bend right] node {$\frac{\e^3}{3}$} (y)	
		(z) edge [loop above] node {$\frac{1}{2}$} ()
		 edge [bend left] node {$\frac{1}{4}$} (x)
		 edge [bend left] node {$\frac{1}{4}$} (y)
		(y) edge [loop below] node {$1-\e$} ()
		edge [bend left] node {$\e$} (z);
 \end{tikzpicture}};

\node (b) [right of = a]{
\begin{tikzpicture}[shorten >=1pt,node distance=2cm, auto]
  \node[state] (x) {$x$};
  \node[state] (z) [above right of = x] {$z$};
  \node[state] (y) [below right of = z] {$y$};

\path[->] (x) edge [loop below] node {$\eqclass{1}$} ()
		edge [bend left] node {$\eqclass{\e^2}$} (z)
		edge [bend right] node {$\eqclass{\e^3}$} (y)	
		(z) edge [loop above] node {$\eqclass{1}$} ()
		 edge [bend left] node {$\eqclass{1}$} (x)
		 edge [bend left] node {$\eqclass{1}$} (y)
		(y) edge [loop below] node {$\eqclass{1}$} ()
		edge [bend left] node {$\eqclass{\e}$} (z);
 \end{tikzpicture}
};
 
\node (c) [below of = b] {
\begin{tikzpicture}[shorten >=1pt,node distance=2cm, auto]
  \node[state] (x) {$x$};
  \node[state] (y) [right of = x] {$y$};

\path[->] (x) edge [loop below] node {$\eqclass{1}$} ()
		edge [bend left] node {$\eqclass{\e^2}$} (y)
		(y) edge [loop below] node{$\eqclass{1}$} ()
		edge [bend left] node {$\eqclass{\e}$} (x);
 \end{tikzpicture}};
 
\node (d) [left of = c]{
\begin{tikzpicture}[shorten >=1pt,node distance=3cm, auto]
  \node[state] (x) {$x$};
  \node[state] (y) [right of = x] {$y$};

\path[->] (x) edge [loop below] node {$1 - \dots$} ()
		edge [bend left] node {$\frac{\max(\frac{3}{4}\e^2,4\e^3)}{12-3\e^2-4\e^3 + \max(\frac{3}{4}\e^2,4\e^3)}$} (y)
		(y) edge [loop below] node {$\frac{ \max(\e,4-4\e)}{\e + \max(\e,4-4\e)}$} ()
		edge [bend left] node [above] {$\frac{\e}{\e + \max(\e,4-4\e)}$} (x);
 \end{tikzpicture}};

\draw[->] (a) to node {$\eqclass{\,\,}$} (b);
\draw[->] (b) to node {$\eqclass{\chi}$} (c);
\draw[->] (a) to node {$\delta$} (d);
\draw[->] (d) to node {$\eqclass{\,\,}$} (c);
 \end{tikzpicture}

\item Let $x,y\in S\setminus E_i$. First, $\eqclass{\kappa}(\eqclass{p},E_i)(x,y) = \eqclass{p}(x,y) = \eqclass{p(x,y)} = \eqclass{\kappa(p,x_i)(x,y)}$ by definitions of $\eqclass{\kappa}$, $\eqclass{p}$, and $\kappa$. Also $\eqclass{\kappa}(\eqclass{p},E_i)(\cup E_i,y) = \max_{\eqclass{\precsim}}\{\eqclass{p}(x,y)\,|\,x\in E_i\} = \eqclass{\max\{p(x,y)\,|\,x\in E_i\}} = \eqclass{\kappa(p,x_i)(\cup E_i,y)}$ by definition, Lemma~\ref{lem:cong}.\ref{lem:cong6}, and Lemma~\ref{lem:p-cong-max}. Likewise $\eqclass{\kappa}(\eqclass{p},E_i)(y,\cup E_i) = \max_{\eqclass{\precsim}}\{\eqclass{p}(y,x)\,|\,x\in E_i\} = \eqclass{\max\{p(y,x)\,|\,x\in E_i\}} = \eqclass{\kappa(p,x_i)(y,\cup E_i)}$.

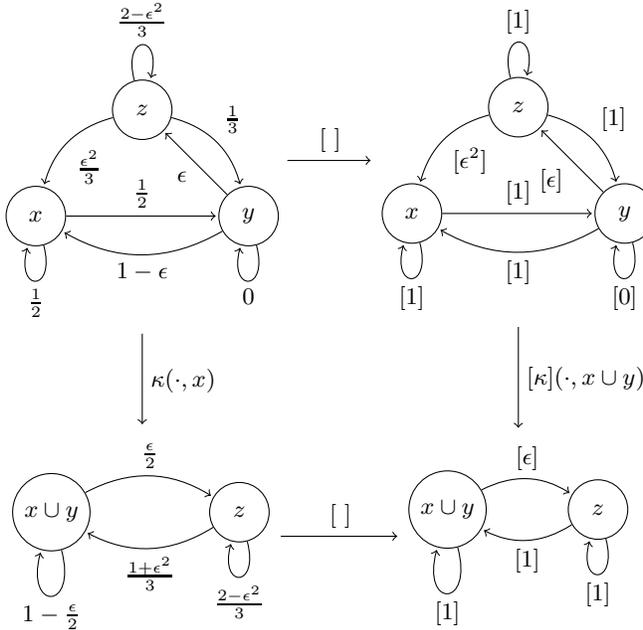
\begin{tikzpicture}[shorten >=1pt,node distance=5cm, auto]

\node (a) {
\begin{tikzpicture}[shorten >=1pt,node distance=2cm, auto]
  \node[state] (x) {$x$};
  \node[state] (z) [above right of = x] {$z$};
  \node[state] (y) [below right of = z] {$y$};

\path[->] (x) edge [loop below] node {$\frac{1}{2}$} ()
		edge node {$\frac{1}{2}$} (y)	
		(z) edge [loop above] node {$\frac{2-\e^2}{3}$} ()
		 edge [bend right] node {$\frac{\e^2}{3}$} (x)
		 edge [bend left] node {$\frac{1}{3}$} (y)
		(y) edge [loop below] node {$0$} ()
		edge node {$\e$} (z)
		edge [bend left] node {$1-\e$} (x);
 \end{tikzpicture}
};

\node (b) [right of = a]{
\begin{tikzpicture}[shorten >=1pt,node distance=2cm, auto]
  \node[state] (x) {$x$};
  \node[state] (z) [above right of = x] {$z$};
  \node[state] (y) [below right of = z] {$y$};

\path[->] (x) edge [loop below] node {$\eqclass{1}$} ()
		edge node {$\eqclass{1}$} (y)	
		(z) edge [loop above] node {$\eqclass{1}$} ()
		 edge [bend right] node {$\eqclass{\e^2}$} (x)
		 edge [bend left] node {$\eqclass{1}$} (y)
		(y) edge [loop below] node {$\eqclass{0}$} ()
		edge node {$\eqclass{\e}$} (z)
		edge [bend left] node {$\eqclass{1}$} (x);
 \end{tikzpicture}};
 
\node (c) [below of = b] {
\begin{tikzpicture}[shorten >=1pt,node distance=2cm, auto]
  \node[state] (x) {$x\cup y$};
  \node[state] (z) [right of = x] {$z$};

\path[->] (x) edge [loop below] node {$\eqclass{1}$} ()
		edge [bend left] node {$\eqclass{\e}$} (z)	
		(z) edge [loop below] node {$\eqclass{1}$} ()
		 edge [bend left] node {$\eqclass{1}$} (x);
\end{tikzpicture}};
 
\node (d) [left of = c]{
\begin{tikzpicture}[shorten >=1pt,node distance=2.5cm, auto]
  \node[state] (x) {$x\cup y$};
  \node[state] (z) [right of = x] {$z$};

\path[->] (x) edge [loop below] node {$1-\frac{\e}{2}$} ()
		edge [bend left] node {$\frac{\e}{2}$} (z)	
		(z) edge [loop below] node {$\frac{2-\e^2}{3}$} ()
		 edge [bend left] node {$\frac{1+\e^2}{3}$} (x);
\end{tikzpicture}};

\draw[->] (a) to node {$\eqclass{\,\,}$} (b);
\draw[->] (b) to node {$\eqclass{\kappa}(\cdot,x\cup y)$} (c);
\draw[->] (a) to node {$\kappa(\cdot,x)$} (d);
\draw[->] (d) to node {$\eqclass{\,\,}$} (c);
 \end{tikzpicture}

\item Let $P := \eqclass{p}$ and $\leq := \eqclass{\precsim}$, and let us prove the claim abstractly. First note that $\eqclass{\chi}\circ\eqclass{\kappa}(P,E_1)$ and $\eqclass{\chi}(P)$ have the same state space $\{\cup E_1,\dots,\cup E_k\}$. For $i,j\neq 1$ the definition of $\eqclass{\chi}$ gives $\eqclass{\chi}\circ\eqclass{\kappa}(P,E_1)(\cup E_j,\cup E_j) = \max_{\leq}\{\eqclass{\kappa}(P,E_1)(\gamma) : \gamma\in \Gamma_T(E_i,E_j)\}$, where $T := S\setminus \cup_iE_i$. It is equal to $\eqclass{\chi}(P)(\cup E_j,\cup E_j)$ since $\eqclass{\kappa}(P,E_1)(x,y) = P(x,y)$ for all $x,y\in S\setminus E_1$, which then also holds for paths $\gamma\in\Gamma_T(E_i,E_j)$. Let us now show that $\eqclass{\chi}(P)(\cup E_1,\cup E_j) = \eqclass{\chi}\circ\eqclass{\kappa}(P,E_1)(\cup E_1,\cup E_j)$. On the one hand for all paths $x\gamma\in \Gamma_T(E_1,E_j)$ we have $P(x\gamma) \leq \eqclass{\kappa}(P,E_1)((\cup E_1)\gamma)$ since $\eqclass{\kappa}(P,E_1)(\cup E_1,y) := \max_{\leq}\{P(x,y) : x\in E_1\}$, and on the other hand for every $\gamma\in T^*\times E_j$ there exists $x\in E_1$ such that $\eqclass{\kappa}(P,E_1)((\cup E_1)\gamma) = P(x\gamma)$. So 
\begin{align*}
\eqclass{\chi}(P)(\cup E_1,\cup E_j) & = \max_{\leq}\{P(\gamma) : \gamma\in\Gamma_T(E_1,E_j)\}\textrm{ by definition},\\
	& = \max_{\leq}\{\eqclass{\kappa}(P,E_1)(\gamma) : \gamma\in \Gamma_T(\cup E_1,E_j)\}\\
	&\textrm{ by the remark above},\\
	& = \eqclass{\chi}\circ\eqclass{\kappa}(P,E_1)(\cup E_1,\cup E_j) \textrm{ by definition}.
\end{align*}
The equality $\eqclass{\chi}(P)(\cup E_i,\cup E_1) = \eqclass{\chi}\circ\eqclass{\kappa}(P,E_1)(\cup E_i,\cup E_1)$ can be proved likewise.

\item Let us first prove $\eqclass{\delta\circ\kappa(\dots\kappa(\kappa(p,x_1),x_2)\dots,x_k)} (\cup E_i,\cup E_j) = \eqclass{\chi}(\eqclass{p})(\cup E_i,\cup E_j)$ for all $i \neq j$ by induction on the number $k'$ of non-singleton essential classes. Since collapsing a singleton class has no effect, the claim holds for $k' = 0$ by Lemma~\ref{lem:pop}.\ref{lem:pop2}, so let us assume that it holds for some arbitrary $k'$ and that $p$ has $k'+1$ non-singleton essential classes. One may assume up to commuting and renaming that $E_1$ is not a singleton. Since $\kappa(\kappa(p,x_1),\cup E_1) = \kappa(p,x_1)$, also $\delta\circ\kappa(\dots\kappa(\kappa(p,x_1),x_2)\dots,x_{k}) = \delta\circ\kappa(\dots\kappa(\kappa(\kappa(p,x_1),\cup E_1),x_2)\dots,x_{k})$. So $\eqclass{\delta\circ\kappa(\dots\kappa(\kappa(p,x_1),x_2)\dots,x_{k})} (\cup E_i,\cup E_j) =$\\
$ \eqclass{\chi}(\eqclass{\kappa(p,x_1)})(\cup E_i,\cup E_j)$ for all $i \neq j$ by induction hypothesis. Moreover, $\eqclass{\chi}(\eqclass{\kappa(p,x_1)}) = \eqclass{\chi}(\eqclass{\kappa}(\eqclass{p},E_1)) = \eqclass{\chi}(\eqclass{p})$ by Lemmas~\ref{lem:pop}.\ref{lem:pop3} and \ref{lem:pop}.\ref{lem:pop4}. Therefore $\eqclass{\chi}(\eqclass{\sigma(p)}) (\cup E_i,\cup E_j)= \eqclass{\delta\circ\kappa(\dots\kappa(\kappa(\sigma(p),x_1),x_2)\dots,x_k)}(\cup E_i,\cup E_j)$ for all $i \neq j$ by Lemma~\ref{lem:pop}.\ref{lem:pop1} and Observation~\ref{obs:ao-refl-irrel}.

\end{enumerate}
\end{proof}

\begin{proof}[Proposition~\ref{prop:complexity}]
Line~\ref{line:type-vertex} from Algorithm~\ref{algo:br} is performed once and takes $n$ steps; Line~\ref{line:type-edge} takes one step and is performed $n^2$ times. Let us now focus on the recursive function \KwHubRec. If all the arcs of the input are labelled with $0$, the algorithm terminates; if not, $p(s,t) = 1$ at least for some distinct $s,t\in S$ after the outgoing scaling at Line~\ref{line:norm-label}, so either the strongly connected component of $s$ is not a sink, or $s$ is in the same strongly connected component as $t$, which implies in both cases that there are fewer $\cup S_i$ than vertices in $S$, and subsequently that \KwHubRec is recursively called at most $n$ times for an input with $n$ vertices.  Lines~\ref{line:max-label}, \ref{line:norm-label}, \ref{line:max-label-edges}, \ref{line:partial-max}, \ref{line:transient}, \ref{line:reduce-transient}, and \ref{line:remove} take at most $O(n^2)$ steps at each call, thus contributing $O(n^3)$ altogether. Tarjan's algorithm and its modification both run in $O(|A|+|S|)$ which is bounded by $O(n^2)$, and moreover the arcs from different recursive steps are also different, so the overall contribution of Line~\ref{line:Tarjan-SSCC} is $O(n^2)$. 

Let us now deal with the more complex Lines~\ref{line:Dijkstra} and \ref{line:global-max}. Let $r$ be the number of recursive calls that are made to \KwHubRec, and at the $j$-th call let $T_j$ denote the vertices otherwise named $T$. Since the $(j+1)$-th recursive call does not involve vertices in $T_j$, we obtain $\sum_{j=1}^r|T_j| \leq n$. The loop at Line~\ref{line:global-max} is taken at most $n^2|T_j|$ times during the $j$-th call, which yields an overall contribution of $O(n^3)$. Likewise, since a basic shortest-path algorithms terminates within $O(n^2)$ steps and since it is called $|T_j|$ times during the $j$-th recursive call, Line~\ref{line:Dijkstra}'s overall contribution is $O(n^3)$.
\end{proof}

\begin{proof}[Proposition~\ref{prop:transient-vanish}]
Up to focusing let $p$ satisfy Assumption~\ref{Assum2}. Let $y$ be in the set of the transient states $T$, so there exist $x\notin T$ and $\gamma\in \Gamma_T(y,x)$ in the essential graph. By Assumption~\ref{Assum2} this implies
\[0 < \liminf_{\e\to 0}p_\e(\gamma) \leq \liminf_{\e\to 0}\p_{\e}^y(\tau^+_x < \tau^+_y).\]
On the other hand, let $E$ be the essential class of $x$. Lemma~\ref{lem:p-cong-max} implies that
\begin{align*}
\p^x(\tau^+_y < \tau^+_x) & \leq 1- \p^x(X_{\tau^+_{S\backslash E\cup\{x\}}} = x)\\
& = \sum_{z\notin E}\p^x(X_{\tau^+_{S\backslash E\cup\{x\}}} = z) \cong \sum_{z\notin E}\max_{t\in E}p(t,z)
\end{align*}
Since $E$ is an essential class, $\sum_{z\notin E}\max_{t\in E}p_\e(t,z) \stackrel{\e\to 0}{\longrightarrow} 0$, thus by Lemma~\ref{lem:hsr}
\[\tilde{\mu}_{\e}(y) \leq \frac{\tilde{\mu}_{\e}(y)}{\mu_{\e}(x)} = \frac{\p_{\e}^x(\tau^+_y < \tau^+_x)}{\p_{\e}^y(\tau^+_x < \tau^+_y)}\stackrel{\e\to 0}{\longrightarrow} 0.\]
\end{proof}

\begin{proof}[Corollary~\ref{cor:strong-assum-stable}]
In the procedure underlying Theorem~\ref{thm:stable-states}, only the states that are transient at some point during the run are deleted. By Proposition~\ref{prop:transient-vanish} these are exactly the fully vanishing states.
\end{proof}

\begin{proof}[Lemma~\ref{lem:gen-Young}]
For all $x\in S$ let $q^x_\e := \sum_{T\in \mathcal{T}_x} \prod_{(z,t)\in T}p_\e(z,t)$ and let $q := (q^y_\e)_{y\in S,\e\in I}$. By irreducibility $q^z_\e > 0$ for all $z\in S$ and $\e\in I$, so let $\mu_\e^z := \frac{q^z_\e}{\sum_{y\in S}q^y_\e}$ for all $z\in S$ and $\e\in I$. Let us assume that $\beta^{y} \precsim \beta^{x}$ for all $y\in S$, so by finiteness of $S$ there exist positive $c$ and $\e_0$ such that $\beta^y_\e \leq c \cdot \beta^x_\e$ for all $y\in S$ and $\e < \e_0$. For all $y\in S$ and $\e < \e_0$ we have

\[q^x_\e \geq \beta^x_\e \geq \frac{\beta^y_\e}{c} \geq \frac{1}{c \cdot |\mathcal{T}_y|} \cdot \sum_{T\in \mathcal{T}_y} \prod_{(z,t)\in T}p_\e(z,t) = \frac{q^y_\e}{c \cdot |\mathcal{T}_y|}\]
\noindent Note that $|\mathcal{T}_y| \leq 2^{|S|^2}$ since a spanning tree of a graph is a subset of its arcs, so
\[\mu_\e^x = \frac{q^x_\e}{\sum_{y\in S}q^y_\e} \geq \frac{1}{c \cdot \sum_{y\in S}|\mathcal{T}_y|} \geq \frac{1}{c \cdot |S|\cdot 2^{|S|^2}}\]
which ensures that $\liminf_{\e\to 0} \mu^x_\e > 0$. By the Markov chain tree theorem $\mu\cdot p = \mu$, so $x$ is a stable state.

Conversely, let us assume that $\neg(\beta^{y} \precsim \beta^{x})$ for some $y\in S$, so for all $c,\e > 0$ there exists a positive $\eta < \e$ such that $c \cdot \beta^x_{\eta} < \beta^y_{\eta}$. Let $c,\e > 0$ and let a positive $\eta < \e$ be such that $c \cdot 2^{|S|^2} \cdot \beta^x_{\eta} < \beta^y_{\eta}$, so $c \cdot \mu_\eta^x < \mu_\eta^y$. Since $\mu \leq 1$, it shows that  $\liminf_{\e\to 0} \mu^x_\e = 0$. 
\end{proof}

\begin{proof}[Observation~\ref{obs:gen-pos-ss}]
let $G$ be the graph with arc $(x,y)$ if $p(x,y) > 0$. Let $E'_1,\dots,E'_{k'}$ be the sink (aka bottom) strongly connected components of $G$, so a state is stable for $p$ iff it is stable for one of the $p\mid_{E'_i \times E'_i}$. Since the $p\mid_{E'_i \times E'_i}$ are irreducible perturbations, Lemma~\ref{lem:gen-Young} can be applied, and by Assumption~\ref{Assum1} the weights of the spanning trees are totally preordered, so there are stable states.
\end{proof}

\begin{proof}[Observation~\ref{obs:gen-ss}]
For all $x,y\in S$ let $I_{xy}$ be the support of $p(x,y) : I\to[0,1]$. By Assumption~\ref{Assum1} the $I_{xy}$ are totally ordered by inclusion. Among these sets let $0\subsetneq I_1 \subsetneq \dots \subsetneq I_l \subsetneq I$ be the non-trivial subsets of $I$. Up to focusing on a smaller neighborhood of $0$ inside $I$, let us assume that $0$ is a limit point of $I_{1}$, all the ${I_{i+1}\setminus I_i }$, and $I\setminus I_{l}$. By Lemma~\ref{lem:cong}.\ref{lem:cong8} a state is stable for $p$ iff it is stable for $p\mid_{I_1}$, all the $p\mid_{I_{i+1}\setminus I_i }$, and $p\mid_{I\setminus I_l}$. These restrictions all satisfy the positivity assumption of Observation~\ref{obs:gen-pos-ss}, whose underlying algorithm computes the stable states in $O(n^3)$. Since there are at most $n^2$ restrictions, stability is decidable in $O(n^5)$.
\end{proof}

\begin{proof}[Proposition~\ref{prop:stable-dec}]
By induction on $n := |\{p(x,y)\,\mid\, x\neq y \wedge p(x,y)\neq 0 \wedge \neg(0 < p(x,y))\}|$. If $n =0$, let $G$ be the graph with arc $(x,y)$ if $p(x,y) > 0$. Let $E'_1,\dots,E'_{k'}$ be the sink SCCs of $G$, so a state is stable for $p$ iff it is stable for one of the $p\mid_{E'_i \times E'_i}$. By decidability of $\precsim$ and since the $p\mid_{E'_i \times E'_i}$ are irreducible perturbations, Lemma~\ref{lem:gen-Young} allows us to compute their stable states. 

If $n > 0$ let $p(x,y)$ be a non-zero function with zeros, and let $J$ be its support. If $0$ is not a limit point of $J$ ($I\backslash J$), the stable states of $p$ are the stable states of $p_{I\backslash J }$ ($p_{J}$), which are computable by induction hypothesis. If $0$ is a limit point of both $J$ and $I\backslash J$, by Lemma~\ref{lem:cong}.\ref{lem:cong8} the stable states wrt $p$ are the states that are stable wrt both $p_{J}$ and $p_{I\backslash J}$, and we can use the induction hypothesis for both. 
\end{proof}

\begin{proof}[Observation~\ref{obs:span-tree}]
By induction. More specifically, let us prove that these roots are preserved and reflected by outgoing scaling, essential collapse, and transient deletion.
\begin{itemize}
\item Since the outgoing scaling divides all the coefficients by the same scale $f\in F$, the weights of the spanning trees are all divided by $f^{|S|-1}$, and the order between them is preserved. 
\item Let $E$ be a (sink) SCC of the essential graph of $P$, and let $x,y \in E$. It is easy to see that a spanning tree rooted at $x$ can be modified (only within $E$) into a spanning tree rooted at $y$ that has the same weight. Since the arcs in $E$ do not contribute to the weight, the essential collapse is safe.
\item Let the sink SCCs of $P$ be singletons, and let $\{y\}$ not be one of those, so there exists a path from $y$ to a sink SCC $\{x\}$. Let $T$ be a spanning tree rooted at $y$. Following $T$, let $x'$ be the successor of $x$, so the weight of $(x,x')$ is less than $1$. Let us modify $T$ into $T'$ by letting $y$ lead to the new root $x$ by a path of weight $1$. The weight of $T'$ is greater than that of $T$ by at least the weight of $(x,x')$. This shows that only essential vertices may be the roots of spanning trees of maximum weights. Moreover, let $T$ be a spanning tree of maximum weight, and let $x$ and $y$ be essential vertices such that following $T$ from $x$ leads to $y$ without visiting any other essential vertex. Then this path between $x$ and $y$ must have maximal weight among all paths from $x$ to $y$ that avoid other essential vertices. So the weight of maximal spanning trees after  transient deletion correspond to the weight before deletion.
\end{itemize}
\end{proof}

\end{document}